\documentclass[journal]{IEEEtran}
\IEEEoverridecommandlockouts

\usepackage[utf8]{inputenc}

\let\olddegree\degree  
\let\degree\relax      
\usepackage{gensymb}   
\let\degree\olddegree  

\usepackage{silence}
\usepackage{esvect}

\usepackage[table,xcdraw]{xcolor}

\usepackage[linesnumbered,ruled,vlined]{algorithm2e}

\usepackage{savesym}

\usepackage{physics}
\usepackage{amsmath}
\usepackage{tikz}
\usepackage{mathdots}
\savesymbol{widering}
\usepackage{yhmath}
\usepackage{cancel}
\usepackage{color}
\usepackage{siunitx}
\usepackage{array}
\usepackage{multirow}
\usepackage{amssymb}
\usepackage{gensymb}
\usepackage{tabularx}
\usepackage{extarrows}
\usepackage{booktabs}
\usetikzlibrary{fadings}
\usetikzlibrary{patterns}
\usetikzlibrary{shadows.blur}
\usetikzlibrary{shapes}

\usepackage{textcomp}
\usepackage{gensymb}

\usepackage{caption}
\usepackage{subcaption}

\usepackage{listings}
\usetikzlibrary{quantikz}

\usepackage{graphicx}

\usepackage{optidef}

\usepackage{hyperref}

\usepackage{amsthm}

\newtheorem*{remark}{Remark}

\theoremstyle{definition}

\newtheorem{exmp}{Example}

\newtheorem{theorem}{Theorem}

\title{\huge{Optimized compiler for Distributed Quantum Computing}}

\author{Daniele Cuomo, Marcello~Caleffi~\IEEEmembership{IEEE Senior Member}, Kevin Krsulich, Filippo Tramonto,\\ Gabriele Agliardi, Enrico Prati, Angela~Sara~Cacciapuoti~\IEEEmembership{IEEE Senior Member}
    \thanks{D. Cuomo is with \href{http://www.quantuminternet.it}{Future Communications Laboratory}; Department of Physics \textit{Ettore Pancini}, University of Naples Federico II, 80126 NA, Italy. E-mail: \href{mailto:daniele.cuomo@unina.it}{daniele.cuomo@unina.it}}
	\thanks{K. Krsulich is with IBM Quantum, IBM T.J. Watson Research Center, Yorktown Heights, NY 10598, USA. Email: \href{mailto:kevin.krsulich@ibm.com}{kevin.krsulich@ibm.com}}
	\thanks{F. Tramonto is with Kyndryl Italia Innovation Services S.r.l., Segrate (MI) 20090, Italy. Part of this work was completed while F. Tramonto was with IBM Client Innovation Center S.r.l., Peschiera Borromeo, MI 20068, Italy. Email: \href{mailto:filippo.tramonto@gmail.com}{filippo.tramonto@gmail.com}}	
	\thanks{G. Agliardi is with Dipartimento di Fisica, Politecnico di Milano, Milano 20133, Italy; IBM Italia S.p.A., Segrate, MI 20090, Italy. Email: \href{mailto:gabrielefrancesco.agliardi@polimi.it}{gabrielefrancesco.agliardi@polimi.it}.}	
	\thanks{E. Prati is with Istituto di Fotonica e Nanotecnologie, Consiglio Nazionale delle Ricerche, Milano 20133; National Inter-university Consortium for Telecommunications (CNIT), Parma I-43124, Italy. E-mail: \href{mailto:enrico.prati@cnr.it}{enrico.prati@cnr.it}}	
    \thanks{A.S. Cacciapuoti and M. Caleffi are with \href{http://www.quantuminternet.it}{Future Communications Laboratory}; Department of Electrical Engineering and Information Technology (DIETI); University of Naples Federico II, 80125 NA, Italy. Laboratorio Nazionale di Comunicazioni Multimediali, National Inter-university Consortium for Telecommunications (CNIT), 80126 NA, Italy. E-mail \href{mailto:angelasara.cacciapuoti@unina.it}{angelasara.cacciapuoti@unina.it}, \href{mailto:marcello.caleffi@unina.it}{marcello.caleffi@unina.it}.}
}

\begin{document}

\maketitle








\begin{abstract}
    Practical distributed quantum computing requires the development of efficient compilers, able to make quantum circuits compatible with some given hardware constraints. This problem is known to be tough, even for local computing. 
    Here, we address it on distributed architectures. As generally assumed in this scenario, \textit{telegates} represent the fundamental remote (inter-processor) operations. Each telegate consists of several tasks: i) entanglement generation and distribution, ii) local operations, and iii) classical communications. Entanglement generations and distribution is an expensive resource, as it is time-consuming and fault-prone.  
    To mitigate its impact, we model an optimization problem that combines running-time minimization with the usage of that resource. 
    Specifically, we provide a parametric \texttt{ILP} formulation, where the parameter denotes a time horizon (or time availability); the objective function count the number of used resources. To minimize the time, a binary search solves the subject \texttt{ILP} by iterating over the parameter.
    Ultimately, to enhance the solution space, we extend the formulation, by introducing a predicate that manipulates the circuit given in input and parallelizes telegates' tasks.
\end{abstract}

\section{Introduction}
\label{sec:introduction}
Distributed architectures are envisioned as a long-term solution to provide practical applications of quantum computing \cite{CacCalTaf-19,CuoCalCac-20,VanDev-16}. They offer a physical substrate to scale up horizontally computing resources, rather than relying on vertical scale-up, coming from single hardware advancement. 
On the flip side, by linking distributed quantum processors, several new challenges arise \cite{Kim-08,PirBra-16,DurLamHeu-17,WehElkHan-18,Cas-18,CacCalTaf-19,CuoCalCac-20}. Here, we consider the \textit{compilation problem} \cite{FerCacAmo-21}, which is generally tough to solve, even on single processor, where an \texttt{NP}-hardness proof is available \cite{BotKisMar-18}. 
An ever growing literature is available, with a variety of proposals for local computing \cite{MasFalMos-08,SirSanCol-18,WilBurZul-19,LiDinXie-19,ZulWil-19,ItoRayIma-19,ZhaZheZha-20,KarTezPet-20,MorParRes-21,MarMorRoc-21, BooDoBec-18,FerAmo-21, LinAnsHar-21, Kon-21} and for distributed computing \cite{BeaBriGra-13,ZomHouHou-18,DaeNavZom-20,FerCacAmo-21, NikMohSed-21,RouZomSar-21, DadZomHos-21,DaeNavZom-21, SarZom-21, ZomDavGho-21, SunGupRam-21}.

Even if quantum processors are already available, distributed architectures are at an early stage and must be discussed from several perspective to grasp what we need. A promising forecast to what such architectures will look like is based on \textit{telegates} as fundamental inter-processor operations.
Each telegate relies on several tasks: (i) the generation and distribution of entangled states among different processors, (ii) local operations and (iii) classical communication. Such tasks make the telegate an expensive resource, especially in terms of running time\footnote{Refer to  \cite{ZhoWanZou-20, KraRanHam-21} for the state of the art on experimental implementations.}. As a consequence, they have critical impact on the performance of the overall computation.
In contrast to such a limit, telegates offer 
remarkable opportunities of parallelization. 
In fact, much circuit manipulation is possible to keep computation independent from telegates' tasks. Therefore, we aim to model an optimization problem that embeds such opportunities.




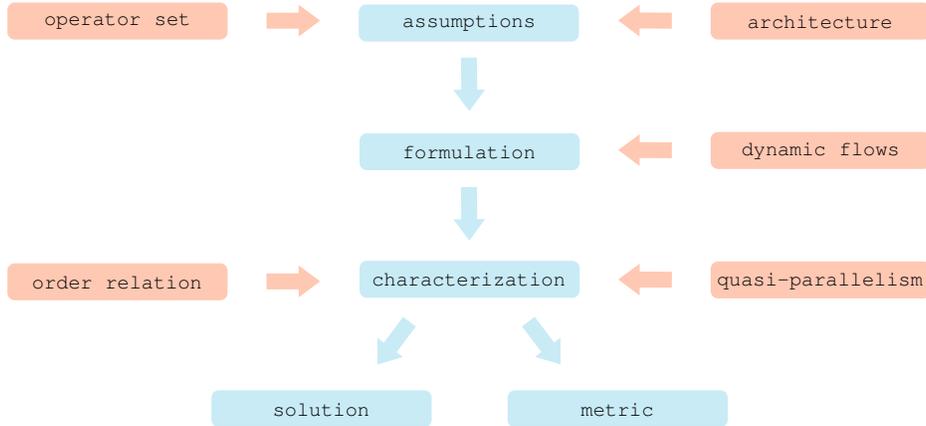
\begin{figure*}
    \centering

\tikzset{every picture/.style={line width=0.75pt}} 

\begin{tikzpicture}[x=0.7pt,y=0.7pt,yscale=-1,xscale=1]

\draw  [color={rgb, 255:red, 0; green, 0; blue, 0 }  ,draw opacity=0 ][fill={rgb, 255:red, 255; green, 71; blue, 0 }  ,fill opacity=0.3 ] (460.75,159.67) .. controls (460.75,157.46) and (462.54,155.67) .. (464.75,155.67) -- (575.75,155.67) .. controls (577.96,155.67) and (579.75,157.46) .. (579.75,159.67) -- (579.75,171.67) .. controls (579.75,173.88) and (577.96,175.67) .. (575.75,175.67) -- (464.75,175.67) .. controls (462.54,175.67) and (460.75,173.88) .. (460.75,171.67) -- cycle ;
\draw  [color={rgb, 255:red, 0; green, 0; blue, 0 }  ,draw opacity=0 ][fill={rgb, 255:red, 255; green, 71; blue, 0 }  ,fill opacity=0.3 ] (439.77,169.84) -- (422.38,169.77) -- (422.36,174.08) -- (410.81,165.41) -- (422.44,156.84) -- (422.42,161.15) -- (439.81,161.22) -- cycle ;
\draw  [color={rgb, 255:red, 0; green, 0; blue, 0 }  ,draw opacity=0 ][fill={rgb, 255:red, 255; green, 71; blue, 0 }  ,fill opacity=0.3 ] (439.77,99.84) -- (422.38,99.77) -- (422.36,104.08) -- (410.81,95.41) -- (422.44,86.84) -- (422.42,91.15) -- (439.81,91.22) -- cycle ;
\draw  [color={rgb, 255:red, 0; green, 0; blue, 0 }  ,draw opacity=0 ][fill={rgb, 255:red, 255; green, 71; blue, 0 }  ,fill opacity=0.3 ] (460.75,89.67) .. controls (460.75,87.46) and (462.54,85.67) .. (464.75,85.67) -- (575.75,85.67) .. controls (577.96,85.67) and (579.75,87.46) .. (579.75,89.67) -- (579.75,101.67) .. controls (579.75,103.88) and (577.96,105.67) .. (575.75,105.67) -- (464.75,105.67) .. controls (462.54,105.67) and (460.75,103.88) .. (460.75,101.67) -- cycle ;
\draw  [color={rgb, 255:red, 0; green, 0; blue, 0 }  ,draw opacity=0 ][fill={rgb, 255:red, 255; green, 71; blue, 0 }  ,fill opacity=0.3 ] (460.75,19.67) .. controls (460.75,17.46) and (462.54,15.67) .. (464.75,15.67) -- (575.75,15.67) .. controls (577.96,15.67) and (579.75,17.46) .. (579.75,19.67) -- (579.75,31.67) .. controls (579.75,33.88) and (577.96,35.67) .. (575.75,35.67) -- (464.75,35.67) .. controls (462.54,35.67) and (460.75,33.88) .. (460.75,31.67) -- cycle ;
\draw  [color={rgb, 255:red, 0; green, 0; blue, 0 }  ,draw opacity=0 ][fill={rgb, 255:red, 255; green, 71; blue, 0 }  ,fill opacity=0.3 ] (439.75,29.77) -- (422.36,29.69) -- (422.34,34) -- (410.79,25.33) -- (422.42,16.76) -- (422.4,21.07) -- (439.79,21.14) -- cycle ;
\draw  [color={rgb, 255:red, 0; green, 0; blue, 0 }  ,draw opacity=0 ][fill={rgb, 255:red, 74; green, 189; blue, 226 }  ,fill opacity=0.3 ] (365.69,185.59) -- (376.16,199.48) -- (379.6,196.88) -- (379.7,211.33) -- (365.83,207.26) -- (369.28,204.67) -- (358.81,190.78) -- cycle ;
\draw  [color={rgb, 255:red, 0; green, 0; blue, 0 }  ,draw opacity=0 ][fill={rgb, 255:red, 74; green, 189; blue, 226 }  ,fill opacity=0.3 ] (301.47,190.83) -- (291.06,204.76) -- (294.51,207.34) -- (280.66,211.46) -- (280.7,197.01) -- (284.15,199.6) -- (294.57,185.67) -- cycle ;
\draw  [color={rgb, 255:red, 0; green, 0; blue, 0 }  ,draw opacity=0 ][fill={rgb, 255:red, 74; green, 189; blue, 226 }  ,fill opacity=0.3 ] (334.45,115.39) -- (334.59,132.78) -- (338.9,132.74) -- (330.38,144.41) -- (321.66,132.88) -- (325.97,132.85) -- (325.83,115.46) -- cycle ;
\draw  [color={rgb, 255:red, 0; green, 0; blue, 0 }  ,draw opacity=0 ][fill={rgb, 255:red, 74; green, 189; blue, 226 }  ,fill opacity=0.3 ] (334.45,45.39) -- (334.59,62.78) -- (338.9,62.74) -- (330.38,74.41) -- (321.66,62.88) -- (325.97,62.85) -- (325.83,45.46) -- cycle ;
\draw  [color={rgb, 255:red, 0; green, 0; blue, 0 }  ,draw opacity=0 ][fill={rgb, 255:red, 74; green, 189; blue, 226 }  ,fill opacity=0.3 ] (191,229.33) .. controls (191,227.12) and (192.79,225.33) .. (195,225.33) -- (306,225.33) .. controls (308.21,225.33) and (310,227.12) .. (310,229.33) -- (310,241.33) .. controls (310,243.54) and (308.21,245.33) .. (306,245.33) -- (195,245.33) .. controls (192.79,245.33) and (191,243.54) .. (191,241.33) -- cycle ;
\draw  [color={rgb, 255:red, 0; green, 0; blue, 0 }  ,draw opacity=0 ][fill={rgb, 255:red, 74; green, 189; blue, 226 }  ,fill opacity=0.3 ] (351,229.33) .. controls (351,227.12) and (352.79,225.33) .. (355,225.33) -- (466,225.33) .. controls (468.21,225.33) and (470,227.12) .. (470,229.33) -- (470,241.33) .. controls (470,243.54) and (468.21,245.33) .. (466,245.33) -- (355,245.33) .. controls (352.79,245.33) and (351,243.54) .. (351,241.33) -- cycle ;
\draw  [color={rgb, 255:red, 0; green, 0; blue, 0 }  ,draw opacity=0 ][fill={rgb, 255:red, 74; green, 189; blue, 226 }  ,fill opacity=0.3 ] (271,159.33) .. controls (271,157.12) and (272.79,155.33) .. (275,155.33) -- (386,155.33) .. controls (388.21,155.33) and (390,157.12) .. (390,159.33) -- (390,171.33) .. controls (390,173.54) and (388.21,175.33) .. (386,175.33) -- (275,175.33) .. controls (272.79,175.33) and (271,173.54) .. (271,171.33) -- cycle ;
\draw  [color={rgb, 255:red, 0; green, 0; blue, 0 }  ,draw opacity=0 ][fill={rgb, 255:red, 74; green, 189; blue, 226 }  ,fill opacity=0.3 ] (271,90.33) .. controls (271,88.12) and (272.79,86.33) .. (275,86.33) -- (386,86.33) .. controls (388.21,86.33) and (390,88.12) .. (390,90.33) -- (390,102.33) .. controls (390,104.54) and (388.21,106.33) .. (386,106.33) -- (275,106.33) .. controls (272.79,106.33) and (271,104.54) .. (271,102.33) -- cycle ;
\draw  [color={rgb, 255:red, 0; green, 0; blue, 0 }  ,draw opacity=0 ][fill={rgb, 255:red, 74; green, 189; blue, 226 }  ,fill opacity=0.3 ] (270.6,20.73) .. controls (270.6,18.52) and (272.39,16.73) .. (274.6,16.73) -- (385.6,16.73) .. controls (387.81,16.73) and (389.6,18.52) .. (389.6,20.73) -- (389.6,32.73) .. controls (389.6,34.94) and (387.81,36.73) .. (385.6,36.73) -- (274.6,36.73) .. controls (272.39,36.73) and (270.6,34.94) .. (270.6,32.73) -- cycle ;
\draw  [color={rgb, 255:red, 0; green, 0; blue, 0 }  ,draw opacity=0 ][fill={rgb, 255:red, 255; green, 71; blue, 0 }  ,fill opacity=0.3 ] (220.75,162.19) -- (238.14,162.14) -- (238.13,157.83) -- (249.75,166.41) -- (238.18,175.07) -- (238.17,170.76) -- (220.78,170.81) -- cycle ;
\draw  [color={rgb, 255:red, 0; green, 0; blue, 0 }  ,draw opacity=0 ][fill={rgb, 255:red, 255; green, 71; blue, 0 }  ,fill opacity=0.3 ] (220.75,21.19) -- (238.14,21.14) -- (238.13,16.83) -- (249.75,25.41) -- (238.18,34.07) -- (238.17,29.76) -- (220.78,29.81) -- cycle ;
\draw  [color={rgb, 255:red, 0; green, 0; blue, 0 }  ,draw opacity=0 ][fill={rgb, 255:red, 255; green, 71; blue, 0 }  ,fill opacity=0.3 ] (80.75,160.67) .. controls (80.75,158.46) and (82.54,156.67) .. (84.75,156.67) -- (195.75,156.67) .. controls (197.96,156.67) and (199.75,158.46) .. (199.75,160.67) -- (199.75,172.67) .. controls (199.75,174.88) and (197.96,176.67) .. (195.75,176.67) -- (84.75,176.67) .. controls (82.54,176.67) and (80.75,174.88) .. (80.75,172.67) -- cycle ;
\draw  [color={rgb, 255:red, 0; green, 0; blue, 0 }  ,draw opacity=0 ][fill={rgb, 255:red, 255; green, 71; blue, 0 }  ,fill opacity=0.3 ] (80.75,19.67) .. controls (80.75,17.46) and (82.54,15.67) .. (84.75,15.67) -- (195.75,15.67) .. controls (197.96,15.67) and (199.75,17.46) .. (199.75,19.67) -- (199.75,31.67) .. controls (199.75,33.88) and (197.96,35.67) .. (195.75,35.67) -- (84.75,35.67) .. controls (82.54,35.67) and (80.75,33.88) .. (80.75,31.67) -- cycle ;

\draw (330.5,96.33) node  [font=\fontsize{0.77em}{0.92em}\selectfont] [align=left] {{\fontfamily{pcr}\selectfont formulation}};
\draw (330.5,165.33) node  [font=\fontsize{0.77em}{0.92em}\selectfont] [align=left] {{\fontfamily{pcr}\selectfont characterization}};
\draw (250.5,235.33) node  [font=\fontsize{0.77em}{0.92em}\selectfont] [align=left] {{\fontfamily{pcr}\selectfont solution}};
\draw (410.5,235.33) node  [font=\fontsize{0.77em}{0.92em}\selectfont] [align=left] {{\fontfamily{pcr}\selectfont metric}};
\draw (520.25,165.67) node  [font=\fontsize{0.77em}{0.92em}\selectfont] [align=left] {{\fontfamily{pcr}\selectfont quasi-parallelism}};
\draw (520.25,95.67) node  [font=\fontsize{0.77em}{0.92em}\selectfont] [align=left] {{\fontfamily{pcr}\selectfont dynamic flows}};
\draw (330.5,26.33) node  [font=\fontsize{0.77em}{0.92em}\selectfont] [align=left] {{\fontfamily{pcr}\selectfont assumptions}};
\draw (520.25,25.67) node  [font=\fontsize{0.77em}{0.92em}\selectfont] [align=left] {{\fontfamily{pcr}\selectfont architecture}};
\draw (140.25,25.67) node  [font=\fontsize{0.77em}{0.92em}\selectfont] [align=left] {{\fontfamily{pcr}\selectfont operator set}};
\draw (140.25,166.67) node  [font=\fontsize{0.77em}{0.92em}\selectfont] [align=left] {{\fontfamily{pcr}\selectfont order relation}};

\end{tikzpicture}
    \caption{Manuscript overview. Blue blocks denote the steps in the problem modeling, scanned by blue arrows. Red blocks are the main ingredients to the entry blue blocks.}
    \label{fig:overview}
\end{figure*}

\subsection{Contribution}
We give a deep analysis on what we can do to mitigate the overhead caused by telegates, which are the main bottleneck to computation on distributed architectures.

Figure \ref{fig:overview} gives a step by step overview of our work, with particular attention to the problem modeling. As reasonable, we begin with some minimal \textbf{assumptions} (Section \ref{sec:essentials}). Namely, as computation model we consider quantum circuits with a universal operator set available. The set is based on local operations and, as said above, telegates as fundamental inter-processor operations. Here, we optimize telegates to efficiently scale with connectivity restrictions.  

We move on by defining rigorously the problem (Section \ref{sec:problem}). To come up with our \textbf{formulation} we rely on a wide literature from the Operation Research field, dealing with network scenarios. Specifically, we noticed several analogies between our problem and those on dynamic networks, especially the group of \textit{multi-commodity flow} problems \cite{FleSku-02,ForFul-58, FulFor-58, SriSta-00, ChoCho-06, CheKhaShe-06, CheKhaShe-04, Sri-97, ChaCheGup-07, CicConPas-21}.
The resulting formulation is particularly remarkable, as it well models an interest into minimizing the running-time, by also keeping as side objective the minimization of resource usage.
In this step, formulation is abstract, as it relies on binary relations that are not fully characterized. We believe that this enhances the modularity of the work and its readability. In fact, exploring the solution space is a combination of resource-availability checks and circuit manipulation, the latter is an hard task on its own and deserves dedicated discussion.
For this reason, the further step is the \textbf{characterization} of those binary relations (Section \ref{sec:relations}). Namely, through these relations, it is possible to discriminate among feasible and unfeasible manipulations. 
First, we relate all the operations to follow the logic induced by the order of occurrence. After that, we also relate operations to discriminate whether they can run in parallel or not. 
This works leads us to generalize the concept of parallelism to a new proposal of ours: the \textit{quasi-parallelism}. This relation is based on (automated) circuit manipulation which aims to gather telegates within the same time step.
The final outcome is a full characterization of the \textbf{solution} space. We evaluate the quality of the solutions available with the quasi-parallelism against no quasi-parallelism -- up to an exponential improvement in the objective function --.
As conclusion, we notice that the final objective function related to a calculated solution results to be a \textbf{metric} to the running-time of computation. This makes our model particularly interesting from a practical perspective.



\begin{table}[b]
\centering
\begin{tabular}{cl}
\hline
{Notation} & {Description} \\ \hline
    $[n]$     &     An enumeration set $\{1, 2, \dots, n\}$        \\
\rowcolor[HTML]{EFEFEF} 
    \texttt{O}      & Font mainly used to denote operators\\
    $\Delta_{\texttt{O}}$    & Time to run operator \texttt{O}\\
    \rowcolor[HTML]{EFEFEF} $\mathcal{N}, \mathcal{Q}$ & Network and quotient graphs\\
    $\mathcal{L}$  & Circuit encoding\\
    \rowcolor[HTML]{EFEFEF} $d$ & Circuit depth\\
    $\tau$ & Discrete time step\\
    \rowcolor[HTML]{EFEFEF} $\star, \prec, \shortparallel$  & Binary relations\\
    $f$ & Flow function\\
    \rowcolor[HTML]{EFEFEF} $q_u$ &  $u$-th computation qubit\\
    $c_u$ &  $u$-th communication qubit\\
    \rowcolor[HTML]{EFEFEF} $b_u$ &  $u$-th classical bit\\
    $P_i$ &  $i$-th quantum processor
\end{tabular}

\end{table}

\section{Distributed quantum computing essentials} 
\label{sec:essentials}
In this section we describe the main elements,  featuring a distributed quantum architecture.

One can encode a quantum processor as a set of qubits and a set of sparse tuneable couplings among qubits. If two qubits are coupled it means that they can interact. We will refer to such couplings as \textit{local couplings}, to emphasize they belong to the same node in distributed architectures, 
as opposed to \textit{entanglement links}, that are couplings between qubits in different processors. As detailed in next sub-section, two remote qubits coupled through an entanglement link, cannot be used for computation: consequently, it is useful to classify qubits as either \textit{computation qubits} or \textit{communication qubits}, respectively. 
While computation qubits process information during the  computation, the communication qubits couple distinct processors through the entanglement.
Figure \ref{fig:3qpu} shows a toy architecture. The purple lines represent the couplings among distributed processors. We refer to such lines as \textit{entanglement links}, as detailed in next sub-section.

\subsection{The entanglement link}
To couple two processors, a communication protocol is necessary, known as \textit{entanglement generation and distribution} \cite{CacCalTaf-19, CacCalVan-20, CuoCalCac-20}. We describe it here as three main steps:
\begin{enumerate}
    \item generating a two-qubits maximally entangled state\footnote{The two-qubits assumption is general and can be extended to multi-qubits protocols.};
    \item splitting the state between distributed processors\footnote{This step implies communication. The interested reader can find in \cite{CacCalVan-20} three different protocols achieving the task.};
    \item storing the partial states in the communication qubits.
\end{enumerate}
When the protocol succeeds, the distributed qubits are correlated and can be exploited to perform non-local operations. For this reason we consider this correlation as a virtual link,
which we refer to as \textit{entanglement link}\footnote{The interested reader can find a discussion about how to achieve practical entanglement generation and distribution, via heralded-based protocols, at Ref. \cite{KozWehVan-21}.}. 
Entanglement links extend the possible interactions to any distributed computation qubits. Specifically, since the communication qubits are locally coupled with computation qubits, with entanglement links one can perform operations between distributed computation qubits, referred to as \textit{telegates}.
More details on the functioning of telegates are reported in Section \ref{sec:rcx_func}. However it is important to keep in mind that, to perform a remote operation, one has to measure the states stored in the communication qubits. As a consequence, an entanglement link is a depletable resource, assigned to a single remote operation. 
After the measurement, a new round of entanglement generation and distribution takes place.

We now give a mathematical description of a distributed architecture, in order to formally describe the functioning of telegates.

\begin{figure}
    \centering

\tikzset{every picture/.style={line width=0.75pt}} 

\begin{tikzpicture}[x=0.75pt,y=0.75pt,yscale=-1,xscale=1]

\draw [color={rgb, 255:red, 189; green, 16; blue, 224 }  ,draw opacity=1 ]   (70,124.25) -- (85,124.25) ;
\draw  [fill={rgb, 255:red, 74; green, 144; blue, 226 }  ,fill opacity=0.5 ][dash pattern={on 3.75pt off 2.25pt on 7.5pt off 1.5pt}] (187.57,99.88) .. controls (187.57,98.11) and (189.01,96.67) .. (190.78,96.67) -- (200.78,96.67) .. controls (202.56,96.67) and (204,98.11) .. (204,99.88) -- (204,109.53) .. controls (204,111.31) and (202.56,112.75) .. (200.78,112.75) -- (190.78,112.75) .. controls (189.01,112.75) and (187.57,111.31) .. (187.57,109.53) -- cycle ;
\draw  [fill={rgb, 255:red, 255; green, 255; blue, 255 }  ,fill opacity=1 ] (71.22,104.5) .. controls (71.22,100.52) and (74.45,97.29) .. (78.43,97.29) .. controls (82.41,97.29) and (85.64,100.52) .. (85.64,104.5) .. controls (85.64,108.48) and (82.41,111.71) .. (78.43,111.71) .. controls (74.45,111.71) and (71.22,108.48) .. (71.22,104.5) -- cycle ;
\draw [color={rgb, 255:red, 189; green, 16; blue, 224 }  ,draw opacity=1 ]   (76.88,48) -- (145.79,47.64) ;
\draw [color={rgb, 255:red, 189; green, 16; blue, 224 }  ,draw opacity=1 ]   (76.88,78) -- (145.79,77.64) ;
\draw  [fill={rgb, 255:red, 74; green, 144; blue, 226 }  ,fill opacity=0.5 ][dash pattern={on 4.5pt off 4.5pt}][line width=0.75]  (27.07,47.27) .. controls (27.07,41.6) and (31.66,37) .. (37.33,37) -- (71.47,37) .. controls (77.14,37) and (81.73,41.6) .. (81.73,47.27) -- (81.73,78.07) .. controls (81.73,83.74) and (77.14,88.33) .. (71.47,88.33) -- (37.33,88.33) .. controls (31.66,88.33) and (27.07,83.74) .. (27.07,78.07) -- cycle ;
\draw  [fill={rgb, 255:red, 255; green, 255; blue, 255 }  ,fill opacity=1 ] (32.46,78) .. controls (32.46,74.02) and (35.69,70.79) .. (39.67,70.79) .. controls (43.65,70.79) and (46.88,74.02) .. (46.88,78) .. controls (46.88,81.98) and (43.65,85.21) .. (39.67,85.21) .. controls (35.69,85.21) and (32.46,81.98) .. (32.46,78) -- cycle ;
\draw    (39.67,55.21) -- (39.67,70.79) ;
\draw  [fill={rgb, 255:red, 255; green, 255; blue, 255 }  ,fill opacity=1 ] (32.46,48) .. controls (32.46,44.02) and (35.69,40.79) .. (39.67,40.79) .. controls (43.65,40.79) and (46.88,44.02) .. (46.88,48) .. controls (46.88,51.98) and (43.65,55.21) .. (39.67,55.21) .. controls (35.69,55.21) and (32.46,51.98) .. (32.46,48) -- cycle ;
\draw  [fill={rgb, 255:red, 255; green, 255; blue, 255 }  ,fill opacity=1 ] (62.46,78) .. controls (62.46,74.02) and (65.69,70.79) .. (69.67,70.79) .. controls (73.65,70.79) and (76.88,74.02) .. (76.88,78) .. controls (76.88,81.98) and (73.65,85.21) .. (69.67,85.21) .. controls (65.69,85.21) and (62.46,81.98) .. (62.46,78) -- cycle ;
\draw  [fill={rgb, 255:red, 255; green, 255; blue, 255 }  ,fill opacity=1 ] (62.46,48) .. controls (62.46,44.02) and (65.69,40.79) .. (69.67,40.79) .. controls (73.65,40.79) and (76.88,44.02) .. (76.88,48) .. controls (76.88,51.98) and (73.65,55.21) .. (69.67,55.21) .. controls (65.69,55.21) and (62.46,51.98) .. (62.46,48) -- cycle ;
\draw    (62.46,48) -- (46.88,48) ;
\draw [color={rgb, 255:red, 189; green, 16; blue, 224 }  ,draw opacity=1 ]   (190.21,47.64) -- (249.79,47.97) ;
\draw  [fill={rgb, 255:red, 74; green, 144; blue, 226 }  ,fill opacity=0.5 ][dash pattern={on 4.5pt off 4.5pt}] (140.4,46.91) .. controls (140.4,41.24) and (145,36.64) .. (150.67,36.64) -- (184.8,36.64) .. controls (190.47,36.64) and (195.07,41.24) .. (195.07,46.91) -- (195.07,77.71) .. controls (195.07,83.38) and (190.47,87.97) .. (184.8,87.97) -- (150.67,87.97) .. controls (145,87.97) and (140.4,83.38) .. (140.4,77.71) -- cycle ;
\draw  [fill={rgb, 255:red, 255; green, 255; blue, 255 }  ,fill opacity=1 ] (145.79,77.64) .. controls (145.79,73.66) and (149.02,70.43) .. (153,70.43) .. controls (156.98,70.43) and (160.21,73.66) .. (160.21,77.64) .. controls (160.21,81.62) and (156.98,84.85) .. (153,84.85) .. controls (149.02,84.85) and (145.79,81.62) .. (145.79,77.64) -- cycle ;
\draw  [fill={rgb, 255:red, 255; green, 255; blue, 255 }  ,fill opacity=1 ] (145.79,47.64) .. controls (145.79,43.66) and (149.02,40.43) .. (153,40.43) .. controls (156.98,40.43) and (160.21,43.66) .. (160.21,47.64) .. controls (160.21,51.62) and (156.98,54.85) .. (153,54.85) .. controls (149.02,54.85) and (145.79,51.62) .. (145.79,47.64) -- cycle ;
\draw  [fill={rgb, 255:red, 255; green, 255; blue, 255 }  ,fill opacity=1 ] (175.79,77.64) .. controls (175.79,73.66) and (179.02,70.43) .. (183,70.43) .. controls (186.98,70.43) and (190.21,73.66) .. (190.21,77.64) .. controls (190.21,81.62) and (186.98,84.85) .. (183,84.85) .. controls (179.02,84.85) and (175.79,81.62) .. (175.79,77.64) -- cycle ;
\draw  [fill={rgb, 255:red, 255; green, 255; blue, 255 }  ,fill opacity=1 ] (175.79,47.64) .. controls (175.79,43.66) and (179.02,40.43) .. (183,40.43) .. controls (186.98,40.43) and (190.21,43.66) .. (190.21,47.64) .. controls (190.21,51.62) and (186.98,54.85) .. (183,54.85) .. controls (179.02,54.85) and (175.79,51.62) .. (175.79,47.64) -- cycle ;
\draw    (183,54.85) -- (183,70.43) ;
\draw  [fill={rgb, 255:red, 74; green, 144; blue, 226 }  ,fill opacity=0.5 ][dash pattern={on 4.5pt off 4.5pt}] (244.4,47.24) .. controls (244.4,41.57) and (249,36.97) .. (254.67,36.97) -- (288.8,36.97) .. controls (294.47,36.97) and (299.07,41.57) .. (299.07,47.24) -- (299.07,78.04) .. controls (299.07,83.71) and (294.47,88.31) .. (288.8,88.31) -- (254.67,88.31) .. controls (249,88.31) and (244.4,83.71) .. (244.4,78.04) -- cycle ;
\draw  [fill={rgb, 255:red, 255; green, 255; blue, 255 }  ,fill opacity=1 ] (249.79,77.97) .. controls (249.79,73.99) and (253.02,70.77) .. (257,70.77) .. controls (260.98,70.77) and (264.21,73.99) .. (264.21,77.97) .. controls (264.21,81.96) and (260.98,85.18) .. (257,85.18) .. controls (253.02,85.18) and (249.79,81.96) .. (249.79,77.97) -- cycle ;
\draw    (257,55.18) -- (257,70.77) ;
\draw  [fill={rgb, 255:red, 255; green, 255; blue, 255 }  ,fill opacity=1 ] (249.79,47.97) .. controls (249.79,43.99) and (253.02,40.77) .. (257,40.77) .. controls (260.98,40.77) and (264.21,43.99) .. (264.21,47.97) .. controls (264.21,51.96) and (260.98,55.18) .. (257,55.18) .. controls (253.02,55.18) and (249.79,51.96) .. (249.79,47.97) -- cycle ;
\draw  [fill={rgb, 255:red, 255; green, 255; blue, 255 }  ,fill opacity=1 ] (279.79,77.97) .. controls (279.79,73.99) and (283.02,70.77) .. (287,70.77) .. controls (290.98,70.77) and (294.21,73.99) .. (294.21,77.97) .. controls (294.21,81.96) and (290.98,85.18) .. (287,85.18) .. controls (283.02,85.18) and (279.79,81.96) .. (279.79,77.97) -- cycle ;
\draw [color={rgb, 255:red, 0; green, 0; blue, 0 }  ,draw opacity=1 ]   (189,124.25) -- (204,124.25) ;
\draw    (175.79,77.64) -- (160.21,77.64) ;
\draw    (279.79,47.97) -- (264.21,47.97) ;
\draw    (62.46,78) -- (46.88,78) ;
\draw    (175.79,47.64) -- (160.21,47.64) ;
\draw    (153,54.85) -- (153,70.43) ;
\draw  [fill={rgb, 255:red, 255; green, 255; blue, 255 }  ,fill opacity=1 ] (279.79,47.97) .. controls (279.79,43.99) and (283.02,40.77) .. (287,40.77) .. controls (290.98,40.77) and (294.21,43.99) .. (294.21,47.97) .. controls (294.21,51.96) and (290.98,55.18) .. (287,55.18) .. controls (283.02,55.18) and (279.79,51.96) .. (279.79,47.97) -- cycle ;
\draw    (287,55.18) -- (287,70.77) ;

\draw (85,124.25) node [anchor=west] [inner sep=0.75pt]  [font=\tiny] [align=left] {\textit{{\fontfamily{pcr}\selectfont \textbf{Entanglement link}}}};
\draw (196.78,104.71) node [anchor=west] [inner sep=0.75pt]  [font=\tiny] [align=left] {\textit{{\fontfamily{pcr}\selectfont \textbf{ \ Processor}}}};
\draw (78.43,104.5) node [anchor=west] [inner sep=0.75pt]  [font=\tiny] [align=left] {\textit{{\fontfamily{pcr}\selectfont \textbf{ \ Qubit}}}};
\draw (54.4,62.67) node  [font=\scriptsize]  {$P_{0}$};
\draw (39.67,48) node  [font=\tiny]  {$q_{1}$};
\draw (69.67,48) node  [font=\tiny]  {$c_{1}$};
\draw (39.67,78) node  [font=\tiny]  {$q_{2}$};
\draw (69.67,78) node  [font=\tiny]  {$c_{2}$};
\draw (167.73,62.31) node  [font=\scriptsize]  {$P_{2}$};
\draw (153,47.64) node  [font=\tiny]  {$c_{3}$};
\draw (183,47.64) node  [font=\tiny]  {$c_{5}$};
\draw (153,77.64) node  [font=\tiny]  {$c_{4}$};
\draw (183,77.64) node  [font=\tiny]  {$q_{3}$};
\draw (271.73,62.64) node  [font=\scriptsize]  {$P_{3}$};
\draw (257,47.97) node  [font=\tiny]  {$c_{6}$};
\draw (257,77.97) node  [font=\tiny]  {$q_{4}$};
\draw (287,77.97) node  [font=\tiny]  {$q_{5}$};
\draw (205,124.25) node [anchor=west] [inner sep=0.75pt]  [font=\tiny] [align=left] {\textit{{\fontfamily{pcr}\selectfont \textbf{Local coupling}}}};
\draw (287,47.97) node  [font=\tiny]  {$q_{6}$};

\end{tikzpicture}
    \caption{Toy distributed quantum architecture with 3 processors.}
    \label{fig:3qpu}
\end{figure}
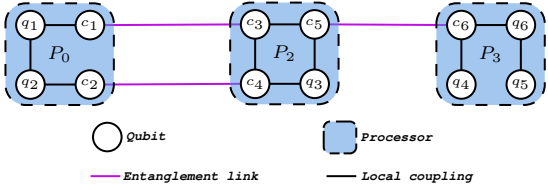

\subsection{Mathematical description}
\label{sec:graph}
So far, we presented the main elements occurring in a distributed quantum architecture, which we can now represent mathematically. Formally, let $\mathcal{N} = (V,P,E)$ be a network triple representing the architecture.
$V = Q \cup C$ is a set of nodes describing qubits, therefore it is the union of computation qubits $Q = \{q_1, q_2,\dots, q_{|Q|}\}$ and communication qubits $C = \{c_1, c_2,\dots, c_{|C|}\}$. We can represent $n$ processors by partitioning $V$ into $P = \{P_1, P_2, \dots, P_n\}$. Therefore, a sub-set $P_i$ characterizes a processor as its set of qubits/nodes.

$E = L \cup R$ is as a set of undirected edges. $L$ represents the local couplings, therefore
\[L\subseteq \bigcup_i P_i\times P_i.\]
Notice that there is no particular assumption on connectivity nor cardinality within processors. This keeps the treating hardware-independent and allows for heterogeneous architectures.

$R$ represents entanglement links. Since entanglement links connect only communication qubits, we introduce, for each processor, a set of those qubits only; i.e. $C_i = C\cap P_i$. Therefore, we have 
\[R \subseteq \bigcup_{i,j\ :\ i \neq j} C_i \times C_j.\]

Figure \ref{fig:3qpu} shows an exemplary architecture, with three processors in $P$, six computation qubits in $Q$, six communication qubits in $C$, three entanglement links in $R$ and ten local couplings in $L$.

As regard minimal assumptions, we only care about architectures actually able to perform any operation. This translated into a simple connection assumption.


\section{Operators}
In the following, the gate model architecture of quantum coputers is considered. There, a circuit describes a time-ordered quantum evolution as a sequence of quantum gates consisting of unitary operators. The set of available operators depends on the physical implementations.





\subsection{Computation operators}
In order to achieve universal quantum computing, one may rely on a universal set of quantum logic gates capable to approximate any possible unitary operator. In the following, we consider a representative universal set of quantum gates, without loss of generality. A sufficient and compact assumption for local universal quantum computing consists of the three-operators set $\{\texttt{H},\texttt{T},\texttt{CX}\}$, where $\texttt{H}$ is the Hadamard operator and $\texttt{T}$ is the $\frac{\pi}{4}$-phase shift. With a polynomial number of repetitions of $\texttt{H}$ and $\texttt{T}$ one can approximate any unitary operator with arbitrary precision \cite{Kit-97, NieChu-02}. 
Other choices of universal sets are possible, such as those based on trapped ions in a cavity \cite{akerman2015universal}, suitable for quantum interfaces where the photonic state is transferred to the cavity mode, and then to the electronic state of the ion via laser pulses \cite{dur2017towards}.


To sum up, the universal operator set for \textbf{local} quantum computing we consider is $\{\texttt{H}, \texttt{T}, \texttt{CX}\}$\footnote{That means a logical circuit will be composed by operators coming from that set only, w.l.o.g.}. 
Whenever an operator occurs with subscript, we are giving information about the qubits it is operating on, e.g., \texttt{CX}$_{u,v}$ is a \texttt{CX} operator with control qubit $q_u$ and target qubit $q_v$. 

\subsection{Universal set}
\label{sec:rcx_func}
To extend the universality also to distributed architectures, we need at least one remote operator. Since the \texttt{CX} is the only operator involving more than one qubit, we just need to implement an operator equivalent to \texttt{CX}, but applying to distributed computation qubits. We call this operator an \texttt{RCX}. Clearly, \texttt{CX} and \texttt{RCX} are equivalent, but with their different nomenclature we highlight their physical difference. Specifically, while \texttt{CX} represents a local gate, \texttt{RCX} represents a sequence of operations that involves distant qubits. Therefore, in general, implementations of \texttt{CX} and \texttt{RCX} come with different fidelity, latency and required resources.

Specifically to the \texttt{RCX} functioning, this is based on a several fundamental steps, which we describe, in turn, by using operators. 
The first operator models the entanglement link creation; we refer to that as \texttt{E} or, more explicitly, as \texttt{E}$_{w,r}$. It sets qubits $c_w$ and $c_r$ to the maximally entangled state $\ket{\Phi^+}$. 
The second operator models a measurement for a communication qubit $c_w$, over the computational basis. Namely, the measurements outputs a classical binary variable $b_w \in \{0,1\}$. 
We refer to that as \texttt{M}$_{w}$ and with the circuit component
\[\begin{quantikz}[thin lines]
    \lstick[]{$c_w$}&\meter[]{}\rstick[]{$b_w$}
\end{quantikz}\]

 \begin{figure*}[ht]
    \centering
    \begin{quantikz}[thin lines,row sep={0.8cm,between origins},transparent]
		\lstick[]{$q_u$}&\qw\gategroup[wires=2,steps=5,style={dashed,rounded corners,inner xsep=11pt,inner ysep=-0.3pt, fill=green!10}, background, label style={rounded corners,label position=above,  yshift=0.08cm, fill=green!10}, background]{$P_i$} & \ctrl{1} & \qw &\qw & \gate[style={fill=green!10}]{\texttt{Z}^{b_r}} &  \qw \\
		\lstick[]{$c_w$}& \gate[2,label style={yshift=0.2cm}]{\texttt{E}} & \targ{} & \qw & \meter[style={fill=green!10}]{} \rstick{$b_w$} \\
		\lstick[]{$c_r$}&\gategroup[wires=2,steps=5,style={dashed,rounded corners,inner xsep=11pt,inner ysep=-0.3pt, fill=red!10}, background,label style={rounded corners,label position=below, yshift=-0.5cm, fill=red!10}, background]{$P_j$} & \ctrl{1} & \gate[style={fill=red!10}]{\texttt{H}} & \meter[style={fill=red!10}]{} \rstick{$b_r$}\\
		\lstick[]{$q_v$} & \qw & \targ{} & \qw & \qw & \gate[style={fill=red!10}]{\texttt{X}^{b_w}} & \qw
	\end{quantikz}
	$\ \ \equiv$
	\begin{quantikz}[thin lines,row sep={0.8cm,between origins}]
        \lstick[]{$q_u$}&\ctrl{1}\gategroup[wires=1,steps=1,style={dashed,rounded corners,inner xsep=11pt,inner ysep=6pt, fill=green!10}, background, label style={rounded corners,label position=above,  yshift=0.08cm, fill=green!10}, background]{$P_i$} & \qw\\
        \lstick[]{$q_v$}& \targ{}\gategroup[wires=1,steps=1,style={dashed,rounded corners,inner xsep=11pt,inner ysep=3pt, fill=red!10}, background,label style={rounded corners,label position=below, yshift=-0.5cm, fill=red!10}, background]{$P_j$} & \qw
    \end{quantikz}
	\caption{Protocol performing an \texttt{RCX}$_{u,v}$. From an operator point of view, this is equivalent to perform $\texttt{CX}_{u,v}$. However $u$ and $v$ belong different processors and that is why we use notation \texttt{RCX}. 
	}
	\label{fig:remop}
\end{figure*}
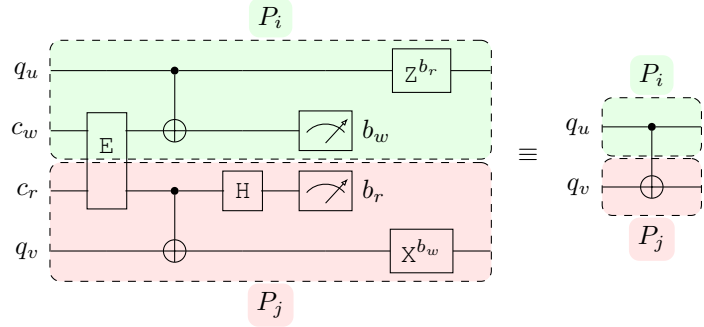
\begin{figure*}
    \centering
    \begin{quantikz}[thin lines,row sep={0.8cm,between origins}, transparent]
		\lstick[]{$c_u$}&\gate[wires=2,label style={yshift=-0.2cm}]{\texttt{E}}\gategroup[wires=1,steps=5,style={dashed,rounded corners,inner xsep=11pt,inner ysep=-0.3pt, fill=green!10}, background, label style={rounded corners,label position=above, xshift =-0.5cm, yshift=0.08cm, fill=green!10}, background]{$P_i$}& \qw & \qw & \qw & \gate[style={fill=green!10}]{\texttt{Z}^{b_v}} & \qw \\
		\lstick[]{$c_v$}& \gategroup[wires=2,steps=5,style={dashed,rounded corners,inner xsep=11pt,inner ysep=-0.3pt, fill=blue!10}, background,label style={rounded corners,label position=below, xshift=0.2cm,yshift=2.45cm, fill=blue!10}, background]{$P_k$} & \ctrl{1} & \gate[style={fill=blue!10}]{\texttt{H}} & \meter[style={fill=blue!10}]{}\rstick[]{$b_v$} \\
		\lstick[]{$c_w$}&\gate[wires=2,label style={yshift=0.2cm}]{\texttt{E}}& \targ{} & \qw & \meter[style={fill=blue!10}]{} \rstick[]{$b_w$}\\
		\lstick[]{$c_r$}&\gategroup[wires=1,steps=5,style={dashed,rounded corners,inner xsep=11pt,inner ysep=-0.3pt, fill=red!10}, background,label style={rounded corners,label position=below,xshift =0.9cm, yshift=3.3cm, fill=red!10}, background]{$P_j$}& \qw & \qw & \qw & \gate[style={fill=red!10}]{\texttt{X}^{b_w}} & \qw
	\end{quantikz}
	$\ \ \equiv$
	\begin{quantikz}[thin lines,row sep={0.8cm,between origins}, transparent]
	    \lstick[]{$c_u$}&\gate[wires=2,label style={yshift=0.2cm}]{\texttt{E}}\gategroup[wires=1,steps=1,style={dashed,rounded corners,inner xsep=11pt,inner ysep=0.25pt, fill=green!10}, background, label style={rounded corners,label position=above, yshift=0.08cm, fill=green!10}, background]{$P_i$} & \qw\\
	    \lstick[]{$c_r$}&\gategroup[wires=1,steps=1,style={dashed,rounded corners,inner xsep=11pt,inner ysep=0.25pt, fill=red!10}, background,label style={rounded corners,label position=below, yshift=-0.5cm, fill=red!10}, background]{$P_j$}& \qw
	\end{quantikz}
    \caption{Entanglement swap protocol. This scenario has three processors $P_i, P_k, P_j$. $P_k$ has an entanglement link both with $P_i$ and with $P_j$, created respectively by $\texttt{E}_{u,v}$ and $\texttt{E}_{w,r}$. At the end of the protocol $c_u$ and $c_r$ are in the maximal entangled state $\ket{\Phi^+}$. From an operator point of view, this is equivalent to perform $\texttt{E}_{u,r}$.}
    \label{fig:swap}
\end{figure*}
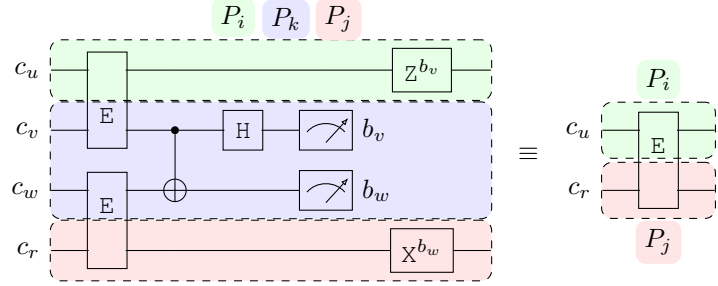

Figure \ref{fig:remop} shows a possible realization of a generic \texttt{RCX}$_{u,v}$.
Here, there are two qubits $q_u,c_w \in P_i$ and two qubits $q_v,c_r \in P_j$. Let us separate the protocol in three different steps.
The first one is the creation of the entanglement link between $c_w$ and $c_r$, i.e. applying \texttt{E}$_{w,r}$. After that, the second step is the \textbf{pre-processing}: a few local operations occur and then qubits $c_w,c_r$ are measured, getting $b_w$ and $b_r$ respectively.
The final step is the \textbf{post-processing}.
The binary variables are used to assert whether further operations are required. Specifically, if $b_r = 1$, a Pauli $\texttt{Z}$ operator applies to $q_u$ and, if $b_w = 1$, a Pauli $\texttt{X}$ operator applies to $q_v$. This phase can be compactly referred with the $\texttt{Z}^{b_r}_{u}, \texttt{X}^{b_w}_{v}$ operators. Notice that $b_w$ is local to processor $P_i$ and $b_r$ is local to $P_j$. But $P_i$ uses $b_r$ and $P_j$ uses $b_w$. In other words, a cross classical communication occurs between $P_i$ and $P_j$.

Let us now give a look to some exemplary applications of \texttt{RCX}$_{u,v}$ over the toy architecture of Figure \ref{fig:3qpu}.

\begin{exmp}
   Assume one wants to run an \texttt{RCX} with control qubit $q_2$ and target $q_3$ -- i.e. \texttt{RCX}$_{2,3}$. Just run circuit in Fig.~\ref{fig:remop}, with $u=2, v=3, w=2, r=4$. 
\end{exmp}

\begin{exmp}
   Now assume one wants to run \texttt{RCX}$_{1,3}$. In this case we can still use the entanglement link between $c_2$ and $c_4$. However, qubit $q_1$ is not coupled with $c_2$. To use that link we need to swap the states stored in $q_1$ and $q_2$ before and after running \texttt{CX}.
   
\end{exmp}

What happens if one wants to run, say, \texttt{RCX}$_{1,4}$? In such a case, the qubits belong two processors having no entanglement link coupling them. There is a really efficient protocol to overcome this problem: it is called \textit{entanglement swap} and we describe it within next section.

\subsection{The entanglement swap}
\label{sec:swap}
As pointed out before, it might be the case where one wants to run an \texttt{RCX} operator between a couple of qubits belonging processors with no entanglement link. Formally, let $P_i$ and $P_j$ such processors and $R \cap (C_i\times C_j) = \emptyset$. In the basic scenario, there exists an intermediate processor $P_k$ which has an entanglement link with both $P_i$ and $P_j$, say via four communication qubits such that $c_u \in C_i$, $c_v,c_w \in C_k$ and $c_r \in C_j$. As Figure \ref{fig:swap} shows, we exploit $P_k$ to entangle $c_u$ and $c_r$.

The entanglement swap protocol can be generalized to an arbitrary sequence of intermediate processors. To this aim we introduce the concept of \textit{entanglement path}.
\subsubsection{The entanglement path}
Coherently with the standard definition of path of a graph, an entanglement path is a sequence of entanglement links connecting two processors. Formally, an entanglement path is a sequence $\{P_{i_1}, P_{i_2}, \dots, P_{i_m}\}$ of $m$ processors such that, $\forall j \in [m-1]$, there is an entanglement link between $P_{i_j}$ and $P_{i_{j+1}}$.

We can therefore entangle two communication qubits $c_u \in P_{i_1}$ and $c_r \in P_{i_m}$ by applying a generalization of the entanglement swap -- showed in Appendix \ref{apx:path} -- to $\{P_{i_1}, P_{i_2}, \dots, P_{i_m}\}$.

Since at the end of the protocol $c_u$ and $c_r$ are in the entangled state $\ket{\Phi^+}$, an entanglement path is a generalization of an entanglement link.

\subsubsection{\texttt{RCX} with entanglement path}
\label{sec:rcx_gen}
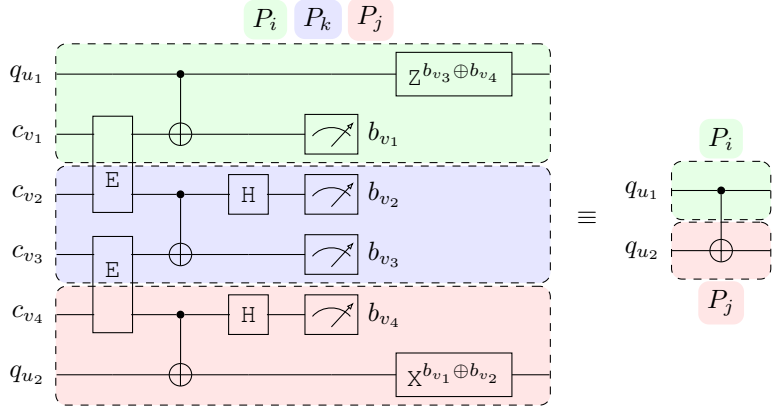
\begin{figure*}
    \centering
    \begin{quantikz}[thin lines,row sep={0.8cm,between origins},transparent]
        \lstick[]{$q_{u_1}$} &\qw\gategroup[wires=2,steps=5,style={dashed,rounded corners,inner xsep=11pt,inner ysep=-0.3pt, fill=green!10}, background, label style={rounded corners,label position=above, xshift =-0.5cm, yshift=0.1cm, fill=green!10}, background]{$P_i$} & \ctrl{1}&\qw & \qw &\gate[]{\texttt{Z}^{b_{v_3} \oplus b_{v_4}}} & \qw\\
		\lstick[]{$c_{v_1}$}&\gate[wires=2,label style={yshift=-0.2cm}]{\texttt{E}}& \targ{} & \qw & \meter[]{}\rstick[]{$b_{v_1}$} \\
		\lstick[]{$c_{v_2}$}& \gategroup[wires=2,steps=5,style={dashed,rounded corners,inner xsep=11pt,inner ysep=-0.3pt, fill=blue!10}, background,label style={rounded corners,label position=below, xshift=0.2cm,yshift=3.3cm, fill=blue!10}, background]{$P_k$}& \ctrl{1} & \gate{\texttt{H}} & \meter{}\rstick[]{$b_{v_2}$}\\
		\lstick[]{$c_{v_3}$}&\gate[wires=2,label style={yshift=0.2cm}]{\texttt{E}}& \targ{} & \qw & \meter{} \rstick[]{$b_{v_3}$}\\
		\lstick[]{$c_{v_4}$}&\gategroup[wires=2,steps=5,style={dashed,rounded corners,inner xsep=11pt,inner ysep=-0.3pt, fill=red!10}, background,label style={rounded corners,label position=below,xshift =0.9cm, yshift=4.95cm, fill=red!10}, background]{$P_j$}& \ctrl{1} & \gate{\texttt{H}} & \meter[]{}\rstick[]{$b_{v_4}$}\\
		\lstick[]{$q_{u_2}$}& \qw & \targ{} & \qw & \qw & \gate[]{\texttt{X}^{b_{v_1}\oplus b_{v_2}}} & \qw
	\end{quantikz}
    $\ \ \equiv$
	\begin{quantikz}[thin lines,row sep={0.8cm,between origins}]
        \lstick[]{$q_{u_1}$}&\ctrl{1}\gategroup[wires=1,steps=1,style={dashed,rounded corners,inner xsep=11pt,inner ysep=6pt, fill=green!10}, background, label style={rounded corners,label position=above,  yshift=0.08cm, fill=green!10}, background]{$P_i$} & \qw\\
        \lstick[]{$q_{u_2}$}& \targ{}\gategroup[wires=1,steps=1,style={dashed,rounded corners,inner xsep=11pt,inner ysep=3pt, fill=red!10}, background,label style={rounded corners,label position=below, yshift=-0.5cm, fill=red!10}, background]{$P_j$} & \qw
    \end{quantikz}
    \caption{\texttt{RCX}$_{u_1, u_2}$ with entanglement swap. 
    }
    \label{fig:rcxswap}
\end{figure*}
In our scenario, the purpose of applying entanglement swap is to perform \texttt{RCX}. For this reason it is interesting to note that we can combine the entanglement swap protocol together with the protocol for \texttt{RCX}. The result is showed in Figure \ref{fig:rcxswap}. 
This result generalizes to every path, no matter the length -- see Appendix \ref{apx:path} --. We further discuss within next section the latency implications coming from this result.



\section{Distributed quantum circuit compilation problem}
\label{sec:problem}
Usually, in the literature dealing with compiler design \cite{ZulWil-19, ItoRayIma-19, WilBurZul-19, FerCacAmo-21}, a circuit is encoded as a set of \textit{layers}. Formally, a layer is a set $\ell$ of independent operators, meaning that each operator in $\ell$ acts on a different collection of qubits. A circuit is an enumeration of layers $\mathcal{L} = \{\ell_1, \ell_2, \dots, \ell_{|\mathcal{L}|}\}$, where the cardinality is also commonly referred as circuit \textit{depth}. 


Usually, a quantum programmer writes a logical circuit, abstracting from the real architecture and assuming that qubits are fully connected, i.e., any couple of qubits can perform a \texttt{CX} operation directly.
Such an abstraction holds also when stepping to distributed architectures\footnote{Recall that, from a user perspective, $\texttt{CX}\equiv \texttt{RCX}$.}. 

However, NISQ architectures do not provide full coupling. As a consequence, there must be a software interface -- namely, a compiler -- able to map abstract circuits to an equivalent one, but meeting the real coupling. In general, such a mapping implies overhead in terms of circuit depth. Therefore, finding a mapping with minimum depth overhead is an optimization problem. We refer to it as the \textit{quantum circuit compilation} problem (\texttt{QCC}), which is proved to be \texttt{NP}-hard \cite{BotKisMar-18}. 
Its version on distributed architectures, 
which we refer to as the \textit{distributed quantum circuit compilation} problem (\texttt{DQCC}), is likely 
to be at least as hard as \texttt{QCC}. In fact, while in \texttt{QCC} we deal with local connectivity restrictions, in \texttt{DQCC} local connectivity stands alongside with remote connectivity -- i.e. the entanglement links --, which is less dense than the local one\footnote{Because the more communication qubits there are, the less computing resources are available.}. Furthermore, performing a remote operation is much more time consuming than a local operation. Just consider that a remote operation relies on communication of both quantum and classical information. 

The above reasons make telegates the bottleneck in distributed computing. Therefore, they are worth of dedicated analysis to minimize their impact. 

\subsection{Objective function}
\label{sec:of}
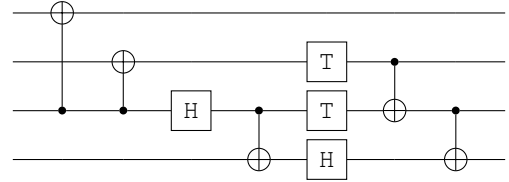
\begin{figure}[hb]
    \centering
    \begin{quantikz}[thin lines,row sep={0.65cm,between origins}, column sep={0.5cm}]
        &\targ{} & \qw& \qw& \qw& \qw& \qw& \qw& \qw\\
        &\qw & \targ{} & \qw &  \qw & \gate{\texttt{T}} & \ctrl{1}& \qw& \qw \\
        & \ctrl{-2} & \ctrl{-1} & \gate[]{\texttt{H}} & \ctrl{1} & \gate[]{\texttt{T}} & \targ{} & \ctrl{1}& \qw\\
        & \qw & \qw & \qw & \targ{} & \gate[]{\texttt{H}} & \qw  & \targ{}& \qw
    \end{quantikz}
    \caption{Exemplary logical circuit, assuming universal gate set $\{\texttt{H}, \texttt{T}, \texttt{CX}, \texttt{RCX}\}$. In the \texttt{DQCC}, part of the instance is a logical circuit, like this one. Depending on the assignment physical-logical qubit, some of the two-qubits operators will be \texttt{CX}, others will be \texttt{RCX}.}
    \label{fig:exmp}
\end{figure}
To optimize a circuit, the first thing we need to do is choosing an objective function to rate the expected performance of a circuit. A common approach is to evaluate only those operators which are somehow a bottleneck to computation. For example, in fault-tolerant quantum computing \cite{Got-98}, this is the $\texttt{T}$ operator \cite{Sel-13,AmyMasMos-14}\footnote{Since main efficient fault-tolerant techniques works only for $\{\texttt{H}, \texttt{CX}\}$.}; in experiments on current technologies, this is the $\texttt{CX}$ operator\footnote{Since it involves interaction of two qubits, it is the more likely operator bringing noise.}. One can count the total number of occurrence of the subject operator $\texttt{O}$, which is the $\texttt{O}$-count; alternatively, one can count the number of layers where at least one operator is $\texttt{O}$, which is the $\texttt{O}$-depth.
To rate a compiled circuit on distributed architectures, we do something along the lines of this latter approach. Specifically, the main bottleneck is the $\texttt{RCX}$ operator and each $\texttt{RCX}$ implies one occurrence of $\texttt{E}$. Therefore, we will rate a circuit with its \texttt{E}-depth.





As simple example of \texttt{E}-depth, consider an instance of the problem: a logical circuit where some \texttt{RCX} operators occur. Figure \ref{fig:exmp} shows an exemplary one. Let us put in the worst-case scenario, i.e., all the four qubits go\footnote{Assigning logical qubits to physical ones is another critical step for compilation and it deserves dedicated analysis \cite{AndHeu-19,DavZomHou-20}, out of the scope of this work.} to different processors. Consequently, all the two-qubits operators are \texttt{RCX}. 
Without considering the tasks which $\texttt{RCX}$ relies on, there is not much optimization to do and the $\texttt{E}$-depth is $5$.

\subsection{Modeling the time domain}
\label{sec:time}
It should be clear that $\texttt{E}$ has central interest in our treating. In fact, we are also going to model the time by scanning it as $\texttt{E}$ occurs.
Specifically, notice that link generations among different couples of qubits are independent. For this reason we assume that all the possible links generate simultaneously and, as soon as all the states are measured, a new round of simultaneous generations begins. Clearly, after that a measurement \texttt{M} generates a boolean $b$, there is at least one post-processing operator that need to wait for that boolean to arrive. Generally speaking, the longer the path the more time $b$ takes to reach its destination.  
We need to account for that by a proper model. To this aim, we do some observations.
\begin{remark}Consider a generic single-qubit unitary operator $\texttt{U}$. The time required to perform $\texttt{U}^b$ is largely dominated by the travel time of $b$, whilst the actual time taken by $\texttt{U}$ can be neglected. Furthermore, the travel of $b$ is independent from computation.
Hence, we can compactly refer to the post-processing waiting-time as $\Delta_{\texttt{U}^b}$. A second observation is that the travel of $b$ is also independent by entanglement link creations, which we assume to take time $\Delta_{\texttt{E}}$. It is also logical to assume $\Delta_{\texttt{U}^b} \lesssim \Delta_{\texttt{E}}$ for the following reasoning: even if $b$ need to cover a longer distance than the one covered by \texttt{E}, $b$ relies on classical technologies, which are way more efficient\footnote{The design of a distributed quantum architecture can easily adapt to satisfy requirements coming from assumptions on classical technologies, since these are very advanced.} than entanglement generation and distribution protocols. 
For this reason, in our treating we neglect $\Delta_{\texttt{U}^b}$, since it happens in parallel with $\Delta_{\texttt{E}}$.
\end{remark}


Stemming from this, we can model the time domain as a discrete set of steps $\tau \in \{1,2, \dots, d\}$, where $d$ is an unknown time horizon, which is also the \texttt{E}-depth. At the beginning of each time step $\tau$, the whole set of entanglement links is available for telegates.
Notice that most of the local operators are expected to run during the creation of the links. Because we relate them to the following inequality
\begin{equation}
\label{eq:time}
    \Delta_{\texttt{E}} \gg \Delta_{\texttt{CX}}, \Delta_{\texttt{H}}, \Delta_{\texttt{T}},
\end{equation}
where, for a generic operator \texttt{O}, $\Delta_{\texttt{O}}$ is the time to run $\texttt{O}$. Therefore, since \texttt{E} is independent from local operators, we can always attempt to run these while \texttt{E} is running -- and also while classical bits $b$ are traveling, as explained in Section \ref{sec:rcx_gen} --.



\subsection{Modeling the distributed architecture}
In light of the above observations, it is reasonable and convenient to consider the whole processor as a network node, and define a function $c$ that provides the number of available links between two processors. Formally, we re-formulate the network graph $\mathcal{N} = (V,P,E)$ introduced in subsection \ref{sec:graph}, to a more compact encoding, which highlights the bottlenecks of a distributed quantum architecture. Specifically, we consider a \textit{quotient graph} of $\mathcal{N}$. To define it, we make use of equivalence classes formalism (it will be useful also later on). Let $\star$ be a an equivalence relation defined on the entanglement links in $R$ as follows:
\begin{equation*}
    (c_u,c_v)\star (c_w,c_r) \iff \exists{i,j}\ :\ c_u,c_w \in C_i\ \land\ c_v,c_r \in C_j.
\end{equation*}

The above statement characterizes the set of inter-partition edges, such that when two edges reach common processors, they belong to the same equivalence class, i.e. 
\[[(c_u,c_v)]_\star = \{(c_w,c_r) \in R\ :\ (c_w,c_r)\star (c_u,c_v)\}.\]
Consequently, the edge set of the quotient graph will be:
\begin{equation}
    R/\star = \{[(c_u,c_v)]_\star\ :\ \forall{(c_u,c_v)} \in R\}.
\end{equation}

The quotient graph is $\mathcal{Q} = (P, R/\star)$. We also associate, to the edges of $\mathcal{Q}$, a capacity function $c: R/\star \rightarrow \mathbb{N}$, such that $c([(c_u,c_v)]_{\star}) = |[(c_u,c_v)]_{\star}|$. It will tell us how many entanglement links are available between two processors.

Ultimately, we define the equivalence class defined on computation qubits, induced by partition $P$, i.e.:
\[q_u P q_v \iff \exists i\ :\ q_u,q_v \in P_i.\]
This is useful to recognize the processor, i.e. the partition, which a generic qubit $q_u$ belongs to, namely $[q_u]_P$.
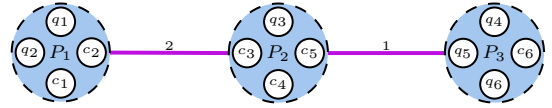
\begin{figure}
    \centering

\tikzset{every picture/.style={line width=0.75pt}} 

\begin{tikzpicture}[x=0.75pt,y=0.75pt,yscale=-1,xscale=1]

\draw [color={rgb, 255:red, 189; green, 16; blue, 224 }  ,draw opacity=1 ][line width=1.5]    (189,192.5) -- (248,192.5) ;
\draw [color={rgb, 255:red, 189; green, 16; blue, 224 }  ,draw opacity=1 ][line width=1.5]    (79,192.25) -- (139,192.5) ;
\draw  [color={rgb, 255:red, 0; green, 0; blue, 0 }  ,draw opacity=1 ][fill={rgb, 255:red, 74; green, 144; blue, 226 }  ,fill opacity=0.5 ][dash pattern={on 4.5pt off 4.5pt}][line width=0.75]  (29.5,192.25) .. controls (29.5,178.58) and (40.58,167.5) .. (54.25,167.5) .. controls (67.92,167.5) and (79,178.58) .. (79,192.25) .. controls (79,205.92) and (67.92,217) .. (54.25,217) .. controls (40.58,217) and (29.5,205.92) .. (29.5,192.25) -- cycle ;
\draw  [fill={rgb, 255:red, 255; green, 255; blue, 255 }  ,fill opacity=1 ] (47.04,176.71) .. controls (47.04,172.73) and (50.27,169.5) .. (54.25,169.5) .. controls (58.23,169.5) and (61.46,172.73) .. (61.46,176.71) .. controls (61.46,180.69) and (58.23,183.92) .. (54.25,183.92) .. controls (50.27,183.92) and (47.04,180.69) .. (47.04,176.71) -- cycle ;
\draw  [color={rgb, 255:red, 0; green, 0; blue, 0 }  ,draw opacity=1 ][fill={rgb, 255:red, 74; green, 144; blue, 226 }  ,fill opacity=0.5 ][dash pattern={on 4.5pt off 4.5pt}] (139,192.5) .. controls (139,178.69) and (150.19,167.5) .. (164,167.5) .. controls (177.81,167.5) and (189,178.69) .. (189,192.5) .. controls (189,206.31) and (177.81,217.5) .. (164,217.5) .. controls (150.19,217.5) and (139,206.31) .. (139,192.5) -- cycle ;
\draw  [color={rgb, 255:red, 0; green, 0; blue, 0 }  ,draw opacity=1 ][fill={rgb, 255:red, 74; green, 144; blue, 226 }  ,fill opacity=0.5 ][dash pattern={on 4.5pt off 4.5pt}] (248,192.5) .. controls (248,178.69) and (259.19,167.5) .. (273,167.5) .. controls (286.81,167.5) and (298,178.69) .. (298,192.5) .. controls (298,206.31) and (286.81,217.5) .. (273,217.5) .. controls (259.19,217.5) and (248,206.31) .. (248,192.5) -- cycle ;
\draw  [fill={rgb, 255:red, 255; green, 255; blue, 255 }  ,fill opacity=1 ] (47.04,207.79) .. controls (47.04,203.81) and (50.27,200.58) .. (54.25,200.58) .. controls (58.23,200.58) and (61.46,203.81) .. (61.46,207.79) .. controls (61.46,211.77) and (58.23,215) .. (54.25,215) .. controls (50.27,215) and (47.04,211.77) .. (47.04,207.79) -- cycle ;
\draw  [fill={rgb, 255:red, 255; green, 255; blue, 255 }  ,fill opacity=1 ] (156.79,176.71) .. controls (156.79,172.73) and (160.02,169.5) .. (164,169.5) .. controls (167.98,169.5) and (171.21,172.73) .. (171.21,176.71) .. controls (171.21,180.69) and (167.98,183.92) .. (164,183.92) .. controls (160.02,183.92) and (156.79,180.69) .. (156.79,176.71) -- cycle ;
\draw  [fill={rgb, 255:red, 255; green, 255; blue, 255 }  ,fill opacity=1 ] (265.79,176.71) .. controls (265.79,172.73) and (269.02,169.5) .. (273,169.5) .. controls (276.98,169.5) and (280.21,172.73) .. (280.21,176.71) .. controls (280.21,180.69) and (276.98,183.92) .. (273,183.92) .. controls (269.02,183.92) and (265.79,180.69) .. (265.79,176.71) -- cycle ;
\draw  [fill={rgb, 255:red, 255; green, 255; blue, 255 }  ,fill opacity=1 ] (265.79,208.29) .. controls (265.79,204.31) and (269.02,201.08) .. (273,201.08) .. controls (276.98,201.08) and (280.21,204.31) .. (280.21,208.29) .. controls (280.21,212.27) and (276.98,215.5) .. (273,215.5) .. controls (269.02,215.5) and (265.79,212.27) .. (265.79,208.29) -- cycle ;
\draw  [fill={rgb, 255:red, 255; green, 255; blue, 255 }  ,fill opacity=1 ] (281.58,192.5) .. controls (281.58,188.52) and (284.81,185.29) .. (288.79,185.29) .. controls (292.77,185.29) and (296,188.52) .. (296,192.5) .. controls (296,196.48) and (292.77,199.71) .. (288.79,199.71) .. controls (284.81,199.71) and (281.58,196.48) .. (281.58,192.5) -- cycle ;
\draw  [fill={rgb, 255:red, 255; green, 255; blue, 255 }  ,fill opacity=1 ] (31.5,192.25) .. controls (31.5,188.27) and (34.73,185.04) .. (38.71,185.04) .. controls (42.69,185.04) and (45.92,188.27) .. (45.92,192.25) .. controls (45.92,196.23) and (42.69,199.46) .. (38.71,199.46) .. controls (34.73,199.46) and (31.5,196.23) .. (31.5,192.25) -- cycle ;
\draw  [fill={rgb, 255:red, 255; green, 255; blue, 255 }  ,fill opacity=1 ] (62.58,192.25) .. controls (62.58,188.27) and (65.81,185.04) .. (69.79,185.04) .. controls (73.77,185.04) and (77,188.27) .. (77,192.25) .. controls (77,196.23) and (73.77,199.46) .. (69.79,199.46) .. controls (65.81,199.46) and (62.58,196.23) .. (62.58,192.25) -- cycle ;
\draw  [fill={rgb, 255:red, 255; green, 255; blue, 255 }  ,fill opacity=1 ] (141,192.5) .. controls (141,188.52) and (144.23,185.29) .. (148.21,185.29) .. controls (152.19,185.29) and (155.42,188.52) .. (155.42,192.5) .. controls (155.42,196.48) and (152.19,199.71) .. (148.21,199.71) .. controls (144.23,199.71) and (141,196.48) .. (141,192.5) -- cycle ;
\draw  [fill={rgb, 255:red, 255; green, 255; blue, 255 }  ,fill opacity=1 ] (156.79,208.29) .. controls (156.79,204.31) and (160.02,201.08) .. (164,201.08) .. controls (167.98,201.08) and (171.21,204.31) .. (171.21,208.29) .. controls (171.21,212.27) and (167.98,215.5) .. (164,215.5) .. controls (160.02,215.5) and (156.79,212.27) .. (156.79,208.29) -- cycle ;
\draw  [fill={rgb, 255:red, 255; green, 255; blue, 255 }  ,fill opacity=1 ] (172.58,192.5) .. controls (172.58,188.52) and (175.81,185.29) .. (179.79,185.29) .. controls (183.77,185.29) and (187,188.52) .. (187,192.5) .. controls (187,196.48) and (183.77,199.71) .. (179.79,199.71) .. controls (175.81,199.71) and (172.58,196.48) .. (172.58,192.5) -- cycle ;
\draw  [fill={rgb, 255:red, 255; green, 255; blue, 255 }  ,fill opacity=1 ] (250,192.5) .. controls (250,188.52) and (253.23,185.29) .. (257.21,185.29) .. controls (261.19,185.29) and (264.42,188.52) .. (264.42,192.5) .. controls (264.42,196.48) and (261.19,199.71) .. (257.21,199.71) .. controls (253.23,199.71) and (250,196.48) .. (250,192.5) -- cycle ;

\draw (54.25,192.25) node  [font=\scriptsize]  {$P_{1}$};
\draw (164,192.5) node  [font=\scriptsize]  {$P_{2}$};
\draw (109,192.38) node [anchor=south] [inner sep=0.75pt]  [font=\tiny]  {$2$};
\draw (273,192.5) node  [font=\scriptsize]  {$P_{3}$};
\draw (218.5,192.5) node [anchor=south] [inner sep=0.75pt]  [font=\tiny]  {$1$};
\draw (54.25,176.71) node  [font=\tiny]  {$q_{1}$};
\draw (38.71,192.25) node  [font=\tiny]  {$q_{2}$};
\draw (164,176.71) node  [font=\tiny]  {$q_{3}$};
\draw (273,176.71) node  [font=\tiny]  {$q_{4}$};
\draw (257.21,192.5) node  [font=\tiny]  {$q_{5}$};
\draw (273,208.29) node  [font=\tiny]  {$q_{6}$};
\draw (54.25,207.79) node  [font=\tiny]  {$c_{1}$};
\draw (69.79,192.25) node  [font=\tiny]  {$c_{2}$};
\draw (148.21,192.5) node  [font=\tiny]  {$c_{3}$};
\draw (164,208.29) node  [font=\tiny]  {$c_{4}$};
\draw (179.79,192.5) node  [font=\tiny]  {$c_{5}$};
\draw (288.79,192.5) node  [font=\tiny]  {$c_{6}$};

\end{tikzpicture}
    \caption{Quotient graph derived from toy network of Figure \ref{fig:3qpu}. The processors become the nodes, the entanglement links between a couple of processors are "compressed" into one edge, with capacity equal to the number of original links. 
    }
    \label{fig:quotient}
	\hrulefill
\end{figure}

In Figure \ref{fig:quotient} we show the quotient graph of the toy architecture of Figure \ref{fig:3qpu}. 



\begin{figure*}[ht]
\centering
   \begin{mini}
    {}{f = \sum_{e \in R/\star}\sum_{i \in [k]}\sum_{\tau \in [d]}{f_{e,i}(\tau)}}{}{}
    \label{mcf}
    \addConstraint{\sum_{e \in \delta^{-}(P_j)}f_{e,i}(\tau) - \sum_{e \in \delta^{+}(P_j)} f_{e,i}(\tau) = 0}{}{\forall i \in [k], \forall{\tau}\in [d], \forall P_j \in P\smallsetminus\{P^C_i,P^T_i\}}
    \addConstraint{\sum_{e \in \delta^{-}(P_i^C)}\sum_{\tau \in [d]} f_{e,i}(\tau) - \sum_{e \in \delta^{+}(P_i^C)}\sum_{\tau\in [d]} f_{e,i}(\tau) = +1\ \ \ }{}{\forall i \in [k]} \addConstraint{\sum_{e \in \delta^{-}(P_i^T)}\sum_{\tau \in [d]} f_{e,i}(\tau) - \sum_{e \in \delta^{+}(P_i^T)}\sum_{\tau\in [d]} f_{e,i}(\tau) = -1\ \ \ }{}{\forall i \in [k]}
    \addConstraint{\sum_{i \in [k]}{f_{e,i}(\tau)} \leq c(e)}{}{\forall e \in R/\star,  \forall \tau \in [d]}  
    \end{mini}
\end{figure*}

\subsection{Single layer formulation}
\label{sec:layer}
Consider a basic circuit expressed as the singleton $\mathcal{L} = \{\ell\}$.
Assume that in $\ell$ there occur $k$ \texttt{RCX} operators. From a logical perspective, all the $k$ operators can run in parallel -- by definition of layer --. 
In other words, if the architecture connectivity had infinite capacity -- i.e. $c(e) = \infty,\ \forall e \in R/\star$ --, we could run $\mathcal{L}$ with \texttt{E}-depth $1$, that is optimal.
As the capacity values decrease, the optimal \texttt{E}-depth value grows, up to \texttt{E}-depth $k$ in the worst-case.
For instance, consider a quotient graph as the one in Figure \ref{fig:2qpu}; it represents a generic 2-processors architecture. In such a simple case, there is not much margin for optimization. Namely, the \texttt{E}-depth depends on the capacity $c(P_1,P_2)$ and the circuit. Take $k$ operations occurring in the same layer. Whenever $c(P_1,P_2) \geq k$, the \texttt{E}-depth is $1$. As the capacity goes below $k$, the depth increases, up to $k$ for $c(P_1,P_2) = 1$.
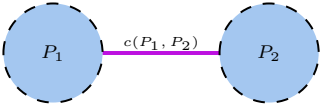
\begin{figure}[t]
    \centering

\tikzset{every picture/.style={line width=0.75pt}} 

\begin{tikzpicture}[x=0.75pt,y=0.75pt,yscale=-1,xscale=1]

\draw [color={rgb, 255:red, 189; green, 16; blue, 224 }  ,draw opacity=1 ][line width=1.5]    (461,193.5) -- (520,193.5) ;
\draw  [color={rgb, 255:red, 0; green, 0; blue, 0 }  ,draw opacity=1 ][fill={rgb, 255:red, 74; green, 144; blue, 226 }  ,fill opacity=0.5 ][dash pattern={on 4.5pt off 4.5pt}] (411,193.5) .. controls (411,179.69) and (422.19,168.5) .. (436,168.5) .. controls (449.81,168.5) and (461,179.69) .. (461,193.5) .. controls (461,207.31) and (449.81,218.5) .. (436,218.5) .. controls (422.19,218.5) and (411,207.31) .. (411,193.5) -- cycle ;
\draw  [color={rgb, 255:red, 0; green, 0; blue, 0 }  ,draw opacity=1 ][fill={rgb, 255:red, 74; green, 144; blue, 226 }  ,fill opacity=0.5 ][dash pattern={on 4.5pt off 4.5pt}] (520,193.5) .. controls (520,179.69) and (531.19,168.5) .. (545,168.5) .. controls (558.81,168.5) and (570,179.69) .. (570,193.5) .. controls (570,207.31) and (558.81,218.5) .. (545,218.5) .. controls (531.19,218.5) and (520,207.31) .. (520,193.5) -- cycle ;

\draw (436,193.5) node  [font=\scriptsize]  {$P_{1}$};
\draw (545,193.5) node  [font=\scriptsize]  {$P_{2}$};
\draw (490.5,193.5) node [anchor=south] [inner sep=0.75pt]  [font=\tiny]  {$c(P_1,P_2)$};

\end{tikzpicture}
    \caption{Generic quotient graph $\mathcal{Q} = (\{P_1, P_2\}, \{(P_1,P_2)\})$ for any two processor architecture.}
    \label{fig:2qpu}
	\hrulefill
\end{figure}
\begin{figure*}[htbp]
\centering
   \begin{mini}
    {}{f = \sum_{e \in R/\star}\sum_{i \in [k]}\sum_{\tau\in [d]}{f_{e,i}(\tau)}}{}{}
    \label{dqcc}
    \addConstraint{\sum_{e \in \delta^{-}(P_j)}f_{e,i}(\tau) - \sum_{e \in \delta^{+}(P_j)} f_{e,i}(\tau) = 0}{}{\forall i \in [k], \forall{\tau}\in [d], \forall P_j \in P\smallsetminus\{P^C_i,P^T_i\}}
    \addConstraint{\sum_{e \in \delta^{-}(P_i^C)}\sum_{\tau \in [d]} f_{e,i}(\tau) - \sum_{e \in \delta^{+}(P_i^C)}\sum_{\tau\in [d]} f_{e,i}(\tau) = +1\ \ \ }{}{\forall i \in [k]} \addConstraint{\sum_{e \in \delta^{-}(P_i^T)}\sum_{\tau \in [d]} f_{e,i}(\tau) - \sum_{e \in \delta^{+}(P_i^T)}\sum_{\tau\in [d]} f_{e,i}(\tau) = -1\ \ \ }{}{\forall i \in [k]}
    \addConstraint{\sum_{i \in [k]}{f_{e,i}(\tau)} \leq c(e)}{}{\forall e \in R/\star, \forall \tau \in [d]}
    \addConstraint{f_{e,i}(\tau) \leq \underset{j \prec i \land j \nshortparallel i}{\text{min}} \sum_{\bar{\tau}< \tau}{f_{e,j}(\bar{\tau})}}{}{\forall i \in [k], \forall e \in \delta^-(P^T_i), \forall \tau \in [d]}
    \addConstraint{f_{e,i}(\tau) \leq \underset{j \prec i \land j \shortparallel i}{\text{min}} \sum_{\bar{\tau}\leq \tau}{f_{e,j}(\bar{\tau})}}{}{\forall i \in [k], \forall e \in \delta^-(P^T_i), \forall \tau \in [d]}
    \end{mini}
\end{figure*}

Let us now formulate an optimization problem for a single-layer, multi-processor architecture -- we will introduce a generalization to any circuit in subsection \ref{sec:layers} --.
Specifically, the \textit{quickest multi-commodity flow} \cite{FleSku-02} wraps that basic scenario. 

In brief, the goal is to find a flow over time which satisfy the constraints imposed by a set of so-called
commodities, which are going to represent the \texttt{RCX} of a quantum circuit. The less time the flow takes, the better. 
To formalize this problem one can directly model an objective function that evaluates a flow by the time it takes. This is an approach employed in \cite{LinJai-14}, but for single commodity. Alternatively, authors in \cite{FleSku-02} propose to start from a formulation of the \textit{multi-commodity flow} problem over time \texttt{MCF}$_d$, where $d$ is a given time horizon\footnote{The choice of using letter $d$ should highlight that the time horizon is going to be the \texttt{E}-depth.}, namely a maximal number of time steps in which the flow is constrained. We prefer this latter way because dynamic flows like \texttt{MCF}$_d$ has been deeply studied since long time ago \cite{ForFul-58, FulFor-58}. Furthermore, this approach has a main drawback, explained at the end of this sub-section, but it doesn't apply to our scenario.

To formulate \texttt{MCF}$_d$, first, we enumerate the occurrences of $\texttt{RCX}$ in $\mathcal{L}$ as a set of commodities $[k] = \{1,2,\dots, k\}$. A set of couples source-sink nodes associates to the commodities. To do that, let $P^C = \{P_1^C, P_2^C, \dots P_k^C\}$ and $P^T = \{P_1^T, P_2^T, \dots P_k^T\}$ be two sets of processors, such that the following holds:
\begin{equation*}
    \texttt{RCX}_{u,v} \in \ell \iff \exists{i}\in [k]\ :\ [q_u]_P = P_i^C, [q_v]_P = P_i^T.
\end{equation*}
Namely, $P^C_i$ ($P^T_i$) is the processor where the control (target) qubit of operation $i$ occurs.

The decision variables of the optimization problem are the time-dependent functions  $f_{e,i}(\tau) \in \{0,1\}$, indicating the flow on edge $e \in R/\star$ dedicated to operation $i \in [k]$ at time $\tau$. The function has a binary co-domain because an operation $i$ uses at most one entanglement link in $e$.

\begin{remark}
When dealing with flows over time, usually a travel-time associates to each edge. Instead, we can assume null travel times, i.e., a flow leaving a source at time $\tau$ reaches the sink at same time $\tau$. The fact that we can model a time-dependent problem in this way is due to the quantum nature of the links, because there is a non-local correlation between linked processors.
This is quite remarkable and may lead new interest into a group of flow problems which are dynamic-static hybrids.
\end{remark}

First, let us introduce the \textit{flow conservation} constraint. Formally, $\forall i\in [k]$, $\forall \tau \in [d]$ and $\forall P_j \in P\smallsetminus\{P^C_i,P^T_i\}$ the following holds:
\begin{equation}
\label{c1}
    \sum_{e \in \delta^{-}(P_j)}f_{e,i}(\tau) - \sum_{e \in \delta^{+}(P_j)} f_{e,i}(\tau) = 0
\end{equation}
where $\delta^-, \delta^+ : P \rightarrow R/\star$ are the standard functions outputting the set of entering and exiting edges of the input node, respectively.

Since a flow $f_{e,i}(\tau) = 1$ identifies the usage of an entanglement link in $e$ to perform $i$, we need to guarantee that the flow going through intermediate links of a path does not stop there. Conversely, whenever an end point of the path occurs in the control or target processor -- i.e. $P^C_i$ or $P^T_i$ --, the \textit{operation demand} -- or \textit{commodity demand} -- constraint holds instead of the conservation constraint.
Namely, $\forall i \in [k]$, this can be written as:

\begin{equation}
\label{c2.1}
\sum_{e \in \delta^{-}(P_i^C)}\sum_{\tau \in [d]} f_{e,i}(\tau) - \sum_{e \in \delta^{+}(P_i^C)}\sum_{\tau \in [d]} f_{e,i}(\tau) = +1
\end{equation}

\begin{equation}
\label{c2.2}
\sum_{e \in \delta^{-}(P_i^T)}\sum_{\tau \in [d]} f_{e,i}(\tau) - \sum_{e \in \delta^{+}(P_i^T)}\sum_{\tau \in [d]} f_{e,i}(\tau) = -1
\end{equation}
The above constraint explicitly requests that a flow dedicate to $i$ reaches its target $P^T_i$, without exiting. Symmetrically, it leaves its control processor $P^C_i$ without returning. Notice that constraint \eqref{c1} forces the operation demand to be satisfied within a single time-step.

The last constraint ensures that, at any time step, the number of operations does not exceed the entanglement resources. Hence, $\forall e \in R/\star$ and $\forall \tau \in [d]$, we introduce a \textit{capacity bound}:
\begin{equation}
    \label{c3}
    \sum_{i \in [k]}{f_{e,i}(\tau)} \leq c(e)
\end{equation}

Ultimately, the objective function is the total flow $f = \sum_{e \in R/\star}\sum_{i \in [k]}\sum_{\tau}{f_{e,i}(\tau)}$.

By gathering the above equations, we obtain the Integer Linear Programming formulation \eqref{mcf}, which models \texttt{MCF}$_d$. 
A flow $f$ perfectly matches a set of entanglement paths used by the telegates. 



Notice that solutions with cycles are in general feasible, but are senseless in our scenario.
By expressing the problem as a minimization of $f$, a solver will avoid any cycle and will try to use as few entanglement links as possible.

Once defined a solver for $\texttt{MCF}_d$, we just need to use it as proposed in \cite{FleSku-02}, namely the solver occurs as sub-routine within a binary research on the minimum time where a feasible solution exists. 
Since the research space is over time, the algorithm is, in general, pseudo-logarithmic. 
Specifically to our case, we already know that the worst solution is where all the operations run in sequence -- i.e. \texttt{E}-depth equal to $k$ --. Therefore, the time horizon is upper-bounded by $k$ and the binary search has $\log{k}$ calls to the sub-routine.
Algorithm \ref{algo:binary} shows the steps.



Unfortunately, standard $\texttt{MCF}_d$ can't catch the whole features of \texttt{DQCC} when considering any $\mathcal{L} = \{\ell_1, \ell_2, \dots, \ell_{|\mathcal{L}|}\}$, we need to consider that operations in $[k]$ are somehow related each other by a logic determined by $\mathcal{L}$. Hence in the following sub-section we are going to model such relations by introducing extra constraints.

\begin{algorithm}
\DontPrintSemicolon
\KwIn{$\mathcal{Q}, [k]$}
\KwOut{$d$}
  $L \leftarrow 1, R \leftarrow k$\;
  \While{$L \leq R$}{
        {
            $\bar{d} \leftarrow \frac{L + R}{2}$\;
            $s \leftarrow$ \texttt{MCF}$_{\bar{d}}(\mathcal{Q}, [k])$\;
            \uIf{$s$ \textbf{is feasible}}
            {
                $d \leftarrow \bar{d}$\;
                $R \leftarrow \bar{d}-1$\;
            }
            \Else
            {
                $L \leftarrow \bar{d} + 1$\;
            }
        }
    }
\caption{Quickest multi-commodity flow}
\label{algo:binary}
\end{algorithm}

\subsection{Any layer formulation}
\label{sec:layers}
As mentioned, the formulation we just gave is not enough to model the \texttt{DQCC} problem to any $\mathcal{L} = \{\ell_1,\ell_2,\dots, \ell_{|\mathcal{L}|}\}$, 
because a circuit generally follows a logic which is related on the order of occurrence given by $\mathcal{L}$. Therefore, even if it might happen that two operations could run in any order, in general this is not true. One needs to define an order relation which is consistent with the logic of the circuit. From an optimization point of view, a critical matter is to choose an order relation which either wraps most of the good solutions or is prone to optimization algorithms. For this reason and for the sake of clarity, we here refer to a generic, irreflexive, order relation $\prec$ defined over $[k]$, without giving it a unique definition.
Formally, for any $i,j \in [k]$, $i\prec j$ means that to run $j$ we need to ensure that $i$ already ran. 

Starting from $\prec$, we can define a constraint to add to formulation \eqref{mcf}.
Namely, $\forall i \in [k], \forall e \in \delta^-(P^T_i)$ the following holds:
\begin{equation}
\label{c4}
    f_{e,i}(\tau) \leq \underset{j \prec i}{\text{min}} \sum_{\bar{\tau}< \tau}{f_{e,j}(\bar{\tau})}
\end{equation}
Notice that the right part of the inequality is a value in $\{0,1\}$ and takes value $1$ only if all the operations logically preceding $i$ already ran.

The formulation now models \texttt{DQCC}. But, within next sub-section, we refine inequality \eqref{c4} to get a better solution space. 


\begin{figure}[hb]
    \centering
    \begin{quantikz}[thin lines,row sep={0.6cm,between origins}]
        &\ctrl{1}\gategroup[wires=2,steps=1,style={dashed,rounded corners,inner xsep=1pt,inner ysep=-0.3pt, fill=red!10}, background, label style={rounded corners,label position=above,  yshift=0.15cm, fill=red!10}, background]{$i$} & \qw & \qw\\
        &\targ{} & \ctrl{1}\gategroup[wires=2,steps=1,style={dashed,rounded corners,inner xsep=1pt,inner ysep=-0.3pt, fill=blue!10}, background, label style={rounded corners,label position=above,  yshift=0.66cm, fill=blue!10}, background]{$j$} & \qw\\
        & \qw & \targ{} & \qw 
    \end{quantikz}
    \caption{\texttt{RCX} in logical conflict as both $i$ and $j$ operate on second qubit. 
    }
    \label{fig:conflict}
\end{figure}
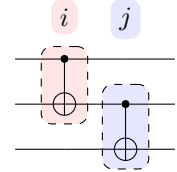

\subsection{Quasi-parallelism}
\label{sec:quasi}
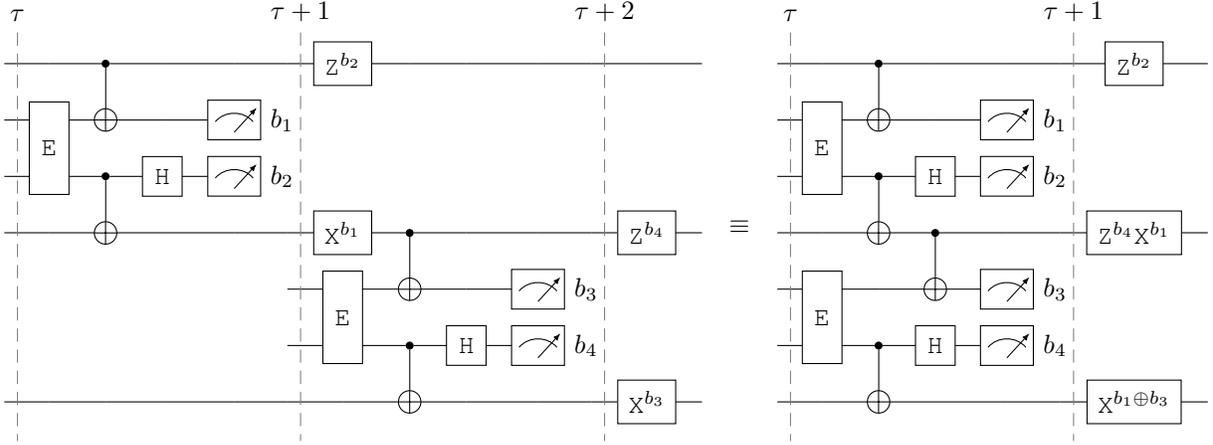
\begin{figure*}
    \centering
    \begin{quantikz}[thin lines,row sep={0.75cm,between origins},column sep=0.35cm]
		\slice[style=gray]{$\tau$}&\qw & \ctrl{1} & \qw & \qw& \qw \slice[style=gray]{$\tau +1$}& \gate{\texttt{Z}^{b_2}} &  \qw & \qw& \qw& \qw\slice[style=gray]{$\tau+2$}& \qw & \qw\\
		& \gate[2]{\texttt{E}} & \targ{} & \qw & \meter{} \rstick{$b_1$} & \\
		& & \ctrl{1} & \gate{\texttt{H}} & \meter{} \rstick{$b_2$} &\\
		& \qw & \targ{} & \qw & \qw & \qw & \gate{\texttt{X}^{b_1}} & \ctrl{1}& \qw& \qw& \qw & \gate{\texttt{Z}^{b_4}}&\qw\\
		& & & & &  &\gate[2]{\texttt{E}} & \targ{} & \qw & \meter{} \rstick{$b_3$} & \\
		& & & & & & & \ctrl{1} & \gate{\texttt{H}} & \meter{} \rstick{$b_4$} &\\
		& \qw& \qw& \qw& \qw& \qw & \qw& \targ{} & \qw & \qw & \qw& \gate{\texttt{X}^{b_3}} & \qw
	\end{quantikz}
	$\ \ \equiv$
    \begin{quantikz}[thin lines,row sep={0.75cm,between origins},column sep=0.35cm]
        \slice[style=gray]{$\tau$}&\qw & \ctrl{1} & \qw &\qw & \qw\slice[style=gray]{$\tau+1$} & \gate{\texttt{Z}^{b_2}} & \qw\\
	    &\gate[wires=2]{\texttt{E}}& \targ{} & \qw & \meter{}\rstick[]{$b_1$} &\\
		&& \ctrl{1} & \gate{\texttt{H}}& \meter{}\rstick[]{$b_2$} & \\
		&\qw &  \targ{} & \ctrl{1} &\qw & \qw & \gate{\texttt{Z}^{b_4}\texttt{X}^{b_1}} & \qw\\
		&\gate[wires=2]{\texttt{E}} & \qw & \targ{}  &\meter{}\rstick[]{$b_3$} &\\
		&& \ctrl{1} & \gate{\texttt{H}} & \meter{}\rstick[]{$b_4$} &\\
		&\qw & \targ{} & \qw & \qw &\qw & \gate{\texttt{X}^{b_1 \oplus b_3}} & \qw
    \end{quantikz}
    \caption{Example of how to achieve quasi-parallelism for two \texttt{RCX} in logical conflict.}
    \label{fig:extended}
\end{figure*}
As before, from an optimization point of view, we are interested in considering as many good solutions as possible. To this aim, we propose an interesting approach which should enlarge the space of good solutions. Specifically, we notice that even if two operations $i,j \in [k]$ are such that $i\prec j$, this does not necessarily mean that they must run at different time steps. They, indeed, may run at the same time step and still respecting the logic imposed by $\prec$.

Consider the example from Figure \ref{fig:conflict}. Since operations $i$ and $j$ operates over a common qubit, they are in logical conflict. Hence, it is reasonable to think that $i\prec j$ should hold. 
However, when considering $i$ and $j$ in their extended form -- i.e. where communication qubits are explicit --, we notice that their logical conflict does not map over all the operations involved.
As Figure \ref{fig:extended} shows, the left part of the equivalence is a naive implementation of $i$ followed by $j$, where the extended form completely inherits the logical conflict. Instead, the right part of the equivalence is way more efficient and it is still an implementation of circuit of Figure \ref{fig:conflict}.
As consequence, even if $i$ and $j$ are in logical conflict, they can run at the same time step.
We refer to this property as \textit{quasi-parallelism}. For this reason we introduce a new binary relation between operations in $[k]$, which we refer to with the intuitive symbol $\shortparallel$. As before, we don't give here a unique definition of $\shortparallel$. Specifically, for any $i,j \in [k]$, we write $i \shortparallel j$ to mean that operations $i$ and $j$ can run at the same time step, but we did not fix a criterion to establish when $\shortparallel$ holds. Clearly, operations $i,j \in [k]$ which can run in full parallelism, are a special case of quasi-parallelism and $i \shortparallel j$ holds.

We can now split the constraint \eqref{c4}, by discriminating between operations which can run in quasi-parallelism and the ones which cannot. Formally, $\forall i \in [k], \forall e \in \delta^-(P^T_i)$ we introduce two new constraints
\begin{equation}
\label{c5}
    f_{e,i}(\tau) \leq \underset{j \prec i \land j \nshortparallel i}{\text{min}} \sum_{\bar{\tau}< \tau}{f_{e,j}(\bar{\tau})}
\end{equation}
\begin{equation}
\label{c6}
    f_{e,i}(\tau) \leq \underset{j \prec i \land j \shortparallel i}{\text{min}} \sum_{\bar{\tau}\leq \tau}{f_{e,j}(\bar{\tau})}
\end{equation}

To sum up, we propose \eqref{dqcc} as Integer Linear Programming formulation of the \texttt{DQCC} problem.
Within next section we fix $\prec$ and $\shortparallel$.


\section{Characterizing the binary relations}
\label{sec:relations}
So far we modeled the problem without completely characterizing relations $\prec$ and $\shortparallel$. This gives a lot of freedom in the way one can tackle the problem. Because, in this way, it is easier to explore different relations, which in turn would bring to different solution spaces.

For practical matters, it is appropriate to keep static definition of $\prec$ and $\shortparallel$, meaning that we do not want the relations to change while the solver is running.

Let us start from $\prec$. We want to make this relation coherent with the order of the layers. More formally, Assume $i,j \in [k]$ and let $\ell_n, \ell_m \in \mathcal{L}$ be such that $i \in \ell_n$ and $j \in \ell_m$. The following holds:
\begin{equation}
    i\prec j \iff n < m.
\end{equation}
With such a requirement, whenever $j$ occurs after $i$ in $\mathcal{L}$, to run $j$ at time $\tau$ we need to run $i$ at a time $\bar{\tau} \leq \tau$ when $i\shortparallel j$ and $\bar{\tau} < \tau$ when $i\nshortparallel j$.
We can now show how impactful our proposal can be. 

\begin{remark}
Consider the following scenario: each layer of $\mathcal{L}$ has at most one remote operation of $[k]$ and all the operation in $[k]$ are in logical conflict one other.
This is interesting from an analytical perspective, since, if we ignored the optimization opportunity offered by quasi-parallelism \footnote{With constraint \eqref{c4} instead of constraints \eqref{c5} and \eqref{c6}.}, the feasible solution space collapses to a singleton. Namely, the only solution is a sequence, where at each time step, only one remote operation occurs, in accordance with $\prec$. The \texttt{E}-depth is then exactly $k$, which is the worse achievable. Instead, by introducing the quasi-parallelism, as we did in formulation \eqref{dqcc}, the space of feasible solutions expands.
Now consider the case where all the operations of $[k]$ can run in quasi-parallelism -- i.e. $\forall i,j \in [k],\ i \shortparallel j$ --. This means that, up to connectivity availability, a solver may find a solution of \texttt{E}-depth $1$.
\end{remark}

Notice that the above reasoning is independent from how one characterizes $\shortparallel$. Hence the goal now is to model $\shortparallel$ to catch as many solutions as possible, while keeping them feasible to the hardware. With this in mind, we propose the following criterion: given any $i,j$, $i \shortparallel j$ holds whenever $i$ and $j$ can run within a certain ``small enough" time lapse. Specifically, the time lapse depends on the coherence time of communication qubits (encoded by $C$), which are assumed to be much more affected by noise than computing qubits (encoded by $Q$). 
Notice that, when two operations $i,j$ run in quasi-parallelism, the life-time of the employed communication qubits might grow. Therefore, we need to ensure that it does not exceed the coherence time of the entanglement. Formally, let us assume $\Delta_c$ being the coherence time of the entanglement -- hence, it starts from the moment \texttt{E} ends, up to the beginning of the measurements \texttt{M} --.

A complication arises from the fact that $\shortparallel$ is, in general, an \textit{intransitive} relation. 
To understand why this is true, consider the circuit in Figure \ref{fig:cascade_conflict}. In such a scenario we are faced with multiple choices. Namely, running
\begin{enumerate}
    \item all $i,j,k$ at different time steps;
    \item all $i,j,k$ at the same time step;
    \item $i, j$ together and $k$ afterwards;
    \item $i$ only, followed by $j,k$ together.
\end{enumerate}

Case 1) is not of interest, because it is the worst solution and no optimization applies.
Case 2) is the best solution, but it is not necessarily feasible. In fact, for $\Delta_c$ small enough, we are forced to split the operations, as in one of the cases 3) and 4).
This explains the non-transitivity, since $i\shortparallel j$ and $j\shortparallel k$, but $i\nshortparallel k$.

We still need to characterize $\shortparallel$, hence, we introduce a predicate method which aims to bring \texttt{RCX} closer to each other, so that quasi-parallelism is achievable.

\begin{figure}[ht]
    \centering
    \begin{quantikz}[thin lines,row sep={0.6cm,between origins}]
        &\ctrl{1}\gategroup[wires=2,steps=1,style={dashed,rounded corners,inner xsep=1pt,inner ysep=-0.3pt, fill=red!10}, background, label style={rounded corners,label position=above,  yshift=0.15cm, fill=red!10}, background]{$i$} & \qw & \qw & \qw\\
        &\targ{} & \ctrl{1}\gategroup[wires=2,steps=1,style={dashed,rounded corners,inner xsep=1pt,inner ysep=-0.3pt, fill=blue!10}, background, label style={rounded corners,label position=above,  yshift=0.66cm, fill=blue!10}, background]{$j$} & \targ{}\gategroup[wires=2,steps=1,style={dashed,rounded corners,inner xsep=1pt,inner ysep=-0.3pt, fill=green!10}, background, label style={rounded corners,label position=above,  yshift=0.65cm, fill=green!10}, background]{$k$} & \qw\\
        & \qw & \targ{} & \ctrl{-1} & \qw
    \end{quantikz}
    \caption{Three \texttt{RCX} operators in logical conflict.}
    \label{fig:cascade_conflict}
\end{figure}
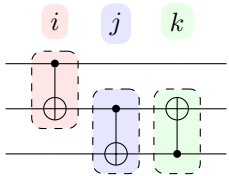
\begin{figure}[ht]
    \centering
        \begin{quantikz}[thin lines,row sep={0.6cm,between origins}]
        &\ctrl{1}\gategroup[wires=2,steps=1,style={dashed,rounded corners,inner xsep=1pt,inner ysep=-0.3pt, fill=red!10}, background, label style={rounded corners,label position=above,  yshift=0.15cm, fill=red!10}, background]{$i$} & \qw & \qw\\
        &\targ{} & \qw & \qw\\
        & \targ{}\gategroup[wires=2,steps=1,style={dashed,rounded corners,inner xsep=1pt,inner ysep=-0.3pt, fill=green!10}, background, label style={rounded corners,label position=above,  yshift=-1.76cm, fill=green!10}, background]{$k$} & \ctrl{1}\gategroup[wires=2,steps=1,style={dashed,rounded corners,inner xsep=1pt,inner ysep=-0.3pt, fill=blue!10}, background, label style={rounded corners,label position=above,  yshift=-1.74cm, fill=blue!10}, background]{$j$} & \qw\\
        & \ctrl{-1} & \targ{} & \qw
    \end{quantikz}
    \caption{Two independent \texttt{RCX} -- i.e. $i$ and $j$ -- belonging different layers.}
    \label{fig:no_conflict}
\end{figure}
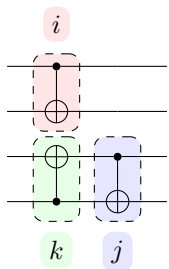

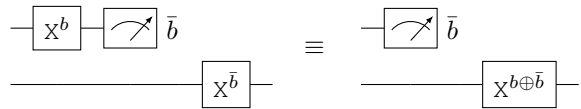
\begin{figure}[t]
    \centering
    \begin{quantikz}[thin lines,row sep={0.75cm,between origins}, column sep ={0.3cm}]
        &\gate[]{\texttt{X}^b}& \meter[]{}\rstick[]{$\bar{b}$}&\\
        &\qw & \qw & \qw & \gate[]{\texttt{X}^{\bar{b}}} & \qw &
    \end{quantikz}
    $\equiv\ $
    \begin{quantikz}[thin lines,row sep={0.75cm,between origins}, column sep ={0.3cm}]
    &\meter[]{}\rstick[]{$\bar{b}$}&\\
    &\qw & \qw &\gate[]{\texttt{X}^{b \oplus \bar{b}}}&\qw &
    \end{quantikz}
    \caption{Propagation of \texttt{X}$^b$. First wire no longer need information of $b$. Second wire need information given by $b\oplus\bar{b}$. Notice that measured $\bar{b}$ is not the same value in the two cases.
    }
    \label{fig:eq}
	\hrulefill
\end{figure}
\begin{figure*}
    \centering
    \begin{quantikz}[thin lines,row sep={0.7cm,between origins},column sep=0.55cm]
        &\ctrl{1} & \qw & \qw & \qw & \qw & \qw\\
        &\targ{} &  \gate[]{\texttt{H}} & \ctrl{1} & \qw & \qw & \qw\\
        & \qw & \qw & \targ{} & \gate[]{\texttt{H}} & \targ{} & \qw\\
         & \qw & \qw & \qw &\qw & \ctrl{-1} & \qw
    \end{quantikz}
    $\ \ \overset{\mathcal{A}}{\longmapsto}$
    \begin{quantikz}[thin lines,row sep={0.7cm,between origins},column sep=0.55cm]
        &\qw & \ctrl{1} & \qw &\qw & \qw & \gate{\texttt{Z}^{b_2}} & \qw\\
	    &\gate[wires=2]{\texttt{E}}& \targ{} & \qw &\qw & \meter{}\rstick[]{$b_1$} &\\
		&& \ctrl{1} & \gate{\texttt{H}}&\qw & \meter{}\rstick[]{$b_2$} & \\
		&\qw &  \targ{} &\gate[]{\texttt{H}} & \ctrl{1}  & \qw & \gate{\texttt{Z}^{b_1\oplus b_4}} & \qw\\
		&\gate[wires=2]{\texttt{E}} & \qw &\qw & \targ{} &\meter{}\rstick[]{$b_3$} &\\
		&& \ctrl{1} & \gate{\texttt{H}} &\qw & \meter{}\rstick[]{$b_4$} &\\
		&\qw & \targ{} & \ctrl{1}  &\gate{\texttt{H}} & \qw& \gate{\texttt{X}^{b_6}\texttt{Z}^{b_3}} & \qw\\
		&\gate[wires=2]{\texttt{E}} & \gate{\texttt{H}} & \targ{} & \qw &\meter{}\rstick[]{$b_5$} &\\
		&& \targ{} & \qw& \qw &\meter{}\rstick[]{$b_6$} &\\
		&\qw & \ctrl{-1} & \qw & \qw &\qw & \gate{\texttt{Z}^{b_5}} & \qw
    \end{quantikz}
    \caption{An expansion, obtained by applying rules from $\mathcal{A}$. In this example scenario, \texttt{RCX} are interspersed with single-qubit  local operators. Notice that boolean variables travel simultaneously. Hence, the assumption we made in Section \ref{sec:time} -- i.e. $\Delta_{\texttt{U}^b} \lesssim \Delta_{\texttt{E}}$ -- holds also for complex evaluations as $\texttt{Z}^{b_1\oplus b_4}$ and $\texttt{X}^{b_6}\texttt{Z}^{b_3}$.
    }
    \label{fig:predicate}
\end{figure*}
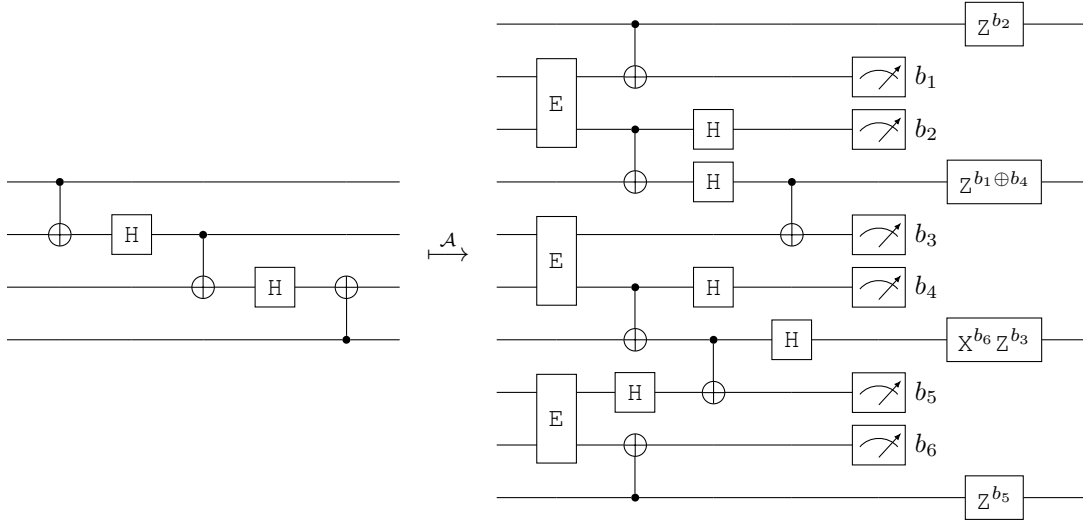
\subsection{Achieving quasi-parallelism, a recursive predicate}
\label{sec:predicate}
As said above, we are now going to introduce a method which verifies if any two telegates can run in quasi-parallelism. Therefore, this method, say $\mathcal{A}(i,j,\Delta_c)$, is a predicate, which is true whenever the operations in input can run in quasi-parallelism. We can finally characterize $\shortparallel$:
\[i \shortparallel j \iff \mathcal{A}(i,j,\Delta_c).\]

$\mathcal{A}$ works in a recursive fashion with three different scenarios as base case. 

\textbf{Base case (i):} given two operations $i,j$, if they belong to the same layer, clearly they can run in full parallelism, therefore $\mathcal{A}(i,j, \Delta_c)$ is true.

\textbf{Base case (ii):} similarly to (i), if $i,j$ belong to different layers and they are completely independent\footnote{Namely, what $i$ does to its qubits does not affect the qubits $j$ operates on.}, $\mathcal{A}(i,j, \Delta_c)$ is true. Circuit of Figure \ref{fig:no_conflict} gives an example with $i,j$ in contiguous layers.

\textbf{Base case (iii):} assume $i,j$ contiguous  -- i.e. in contiguous layers -- and both operating on, at least, one common qubit. We want to introduce, with this base case, the possibility that multiple operators may run simultaneously, as exemplified in Figures~\ref{fig:extended}. For this reason, algorithm $\mathcal{A}$ considers all the operators involved to perform an \texttt{RCX} -- recall protocol from Figure \ref{fig:remop} --.
Namely, $\mathcal{A}$ pushes forward the post-processing of $i$ -- i.e. the Pauli operations \texttt{Z}$^{b}$ and \texttt{X}$^{\bar{b}}$ -- after the pre-processing of $j$ -- i.e. the \texttt{CX} operations --.
One can do that by using the following well known transformation rules:
\begin{itemize}
    \item $\texttt{CX}(\texttt{X}^b \otimes \texttt{I}) \equiv (\texttt{X}^b \otimes \texttt{X}^b)\texttt{CX}$
    \item $\texttt{CX}(\texttt{I} \otimes \texttt{Z}^b) \equiv (\texttt{Z}^b \otimes \texttt{Z}^b)\texttt{CX}$
    \item $\texttt{CX}(\texttt{I} \otimes \texttt{X}^b) \equiv (\texttt{I} \otimes \texttt{X}^b)\texttt{CX}$
    \item $\texttt{CX}(\texttt{Z}^b \otimes \texttt{I}) \equiv (\texttt{Z}^b \otimes \texttt{I})\texttt{CX}$
\end{itemize}

After the application of these rules, some post-processing operation, might have been \textit{propagated} also to communication qubits. Specifically, it may happen that a measurement is preceded by an operation $\texttt{X}^b$. One can always reduce the depth of the circuit by sending $b$ to the target(s) of the measurement. This is indeed what happens in our first example -- Figure \ref{fig:extended} --, where, instead of performing $\texttt{X}^{b_1}$ in the communication qubit, we opt to put it in combination with $\texttt{X}^{b_3}$, achieving a single operation $\texttt{X}^{b_1 \oplus b_3}$ -- see also Figure \ref{fig:eq} for a circuit representation --. At the end of the circuit manipulation, the life-time of the communication qubits may have risen. If it does not exceed $\Delta_c$, then $\mathcal{A}(i,j, \Delta_c)$ is true; otherwise, $\mathcal{A}(i,j, \Delta_c)$ is false.

\textbf{Recursion:} consider now the case where $i$ and $j$ are separated by a sequence of local operations $\texttt{O}_1, \dots, \texttt{O}_n$\footnote{Namely, operations $\texttt{O}_1, \dots, \texttt{O}_n$ belong to layers between the ones of $i$ and $j$}, assumed to be confined to the universal set $\{\texttt{H}, \texttt{T}, \texttt{CX}\}$. In this case, $\mathcal{A}$ applies, recursively, transformations for both $i$ and $j$.
Specifically, as long as possible, it \textit{pushes forward} the post-processing of $i$ by using former rules together with:
\begin{itemize}
    \item $\texttt{T}\texttt{Z}^b \equiv \texttt{Z}^b\texttt{T}$
    \item $\texttt{H}\texttt{X}^b \equiv \texttt{Z}^b\texttt{H}$
\end{itemize}
Ultimately, as long as possible, $\mathcal{A}$ \textit{pushes backward} the pre-processing of $j$ by using the following standard rules:
\begin{itemize}
    \item $\texttt{CX}(\texttt{T} \otimes \texttt{I}) \equiv (\texttt{T} \otimes \texttt{I})\texttt{CX}$
    \item $(\texttt{CX}_{u,v}\otimes\texttt{I})(\texttt{I}\otimes\texttt{CX}_{w,v})\equiv(\texttt{I}\otimes\texttt{CX}_{w,v})(\texttt{CX}_{u,v}\otimes\texttt{I})$
    \item $(\texttt{CX}_{u,v}\otimes\texttt{I})(\texttt{I}\otimes\texttt{CX}_{u,w})\equiv(\texttt{I}\otimes\texttt{CX}_{u,w})(\texttt{CX}_{u,v}\otimes\texttt{I})$
    \item $\texttt{CX}_{u,v}(\texttt{H}\otimes\texttt{H})\equiv (\texttt{H}\otimes\texttt{H})\texttt{CX}_{v,u}$
    \item $(\texttt{I}\otimes\texttt{H})\texttt{CX}_{u,v}(\texttt{H}\otimes\texttt{I})\equiv (\texttt{H}\otimes\texttt{I})\texttt{CX}_{v,u}(\texttt{I}\otimes\texttt{H})$
\end{itemize}
If $\mathcal{A}$ manages to make $i$ post-processing and $j$ pre-processing contiguous, the validity check reduces to the base case scenario. Otherwise, $\mathcal{A}(i,j, \Delta_c)$ is false.

So far, we defined $\mathcal{A}$ only for $i,j$ without any other remote operation in between. Before generalizing the method to any $i$ and $j$ we prove that our definition of $\mathcal{A}$ can be implemented in polynomial time. We need this requirement to keep things tractable.
\begin{theorem}
    $\mathcal{A}$ is polynomial.
    \label{th:a}
\end{theorem}
\begin{proof}
 Assume there occur $n$ local operations, say $\texttt{O}_1, \dots, \texttt{O}_n$, between $i$ and $j$.
 If $\mathcal{A}$ manages to push $i$ forward $\texttt{O}_1$, it means that its post-processing run after $\texttt{O}_1$ and it may only propagate \textit{vertically}, over different qubits -- by construction of the rule set --. As consequence, the depth of the circuit has not increased. Furthermore, the post-processing is still composed by Pauli operations of the kind $\texttt{Z}^{b}$ and $\texttt{X}^{\bar{b}}$. Hence, this holds for any $\texttt{O}_{1\leq\bar{n}\leq n}$ and the recursion is upper-bounded by $\mathcal{O}(n)$.
 
 Symmetrically, if $\mathcal{A}$ manages to push $j$ backward $\texttt{O}_n$, it means that the pre-processing can run before $\texttt{O}_n$. Also in this case, the depth has not increased and the pre-processing is still composed by two independent \texttt{CX} operations -- again, by construction of the rule set --. Hence, this holds for any $\texttt{O}_{1\leq\bar{n}\leq n}$ and the recursion is upper-bounded by $\mathcal{O}(n)$.
\end{proof}

We can now move on to the general case. Formally, between $i$ and $j$ a remote operation $k$ may occur, which is also in logical conflict with both. For such a scenario, we just add a recursive rule. Namely, $\mathcal{A}(i,j,\Delta_c)$ holds iff the following holds:
\[\exists \varepsilon \in [0,1]\ :\ \mathcal{A}(i,k, \varepsilon\cdot \Delta_c) \land \mathcal{A}(k,j,(1-\varepsilon)\cdot\Delta_c).\]
Take a moment to appreciate why this kind of recursion is feasible. Specifically, one might think that validity of $\mathcal{A}(i,k, \varepsilon\cdot\Delta_c)$ and $\mathcal{A}(k,j, (1-\varepsilon)\cdot\Delta_c)$ are not independent, because they both operate on $k$. However, in the former function, $\mathcal{A}$ evaluates the pre-processing of $k$, while, in the latter, it evaluates its post-processing. Therefore they can be evaluated independently.

\begin{theorem}
    Generalized $\mathcal{A}$ is polynomial.
\end{theorem}
\begin{proof}
Assume there occur $k_1,\dots,k_m$ between $i$ and $j$. For the purpose of the proof let $m$ being a power of $2$. $\mathcal{A}(i,j,\Delta_c)$ can choose any of the $k_1,\dots,k_m$ operations for the recursion. To keep symmetry, let $\mathcal{A}(i,k_{\frac{m}{2}}, \varepsilon\cdot\Delta_c)$ and $\mathcal{A}(k_{\frac{m}{2}},j, (1-\varepsilon)\cdot\Delta_c)$ be the recursive call. Notice that operations considered by $\mathcal{A}(i,k_{\frac{m}{2}}, \varepsilon\cdot\Delta_c)$ are $\frac{m}{2}$, as well as the ones considered by $\mathcal{A}(k_{\frac{m}{2}},j, (1-\varepsilon)\cdot\Delta_c)$. The result is a recursive binary tree of height $\log{m}$ and, therefore, $\mathcal{O}(m)$ calls to $\mathcal{A}$. The leaves correspond to the base case of the recursion, which is proved to be tractable in Theorem \ref{th:a}.
    
\end{proof}




Figure \ref{fig:predicate} shows an example scenario where we used rules as in $\mathcal{A}$ -- in addition to the first one of Figure \ref{fig:extended} --.
Clearly, our modular architecture is prone to modifications or extensions of $\mathcal{A}$, if future research highlighted more refined requirements. 

\begin{remark}
Notice that we managed to define $\mathcal{A}$ to be independent by the connectivity of $\mathcal{Q}$. This was possible thanks to the way we modeled telegates via efficient entanglement paths \ref{sec:rcx_gen}. In other words, $\mathcal{A}(i,j,\Delta_c)$ works for any solver and regardless of the path this chooses to perform $i$ and $j$.
As consequence, the characterization of $\mathcal{A}$ -- and therefore also of $\shortparallel$ -- is static and depends only by the logical circuit and global factors, i.e. $\Delta_c$. 
Furthermore, we can relate coherence time and entanglement link creation to $\Delta_{\texttt{E}} + \Delta_c \approx \Delta_{\texttt{E}}$. As consequence, whatever $\Delta_c$ is, $\mathcal{A}$ does not significantly affect the duration of each time step. This makes the \texttt{E}-depth a particularly good index for the running time of the overall computation.
\end{remark}

\section{Discussion}
\subsection{Summary evaluation}
In what follows we evaluate our model for \texttt{DQCC}.



(i) By expressing the problem as a quickest flow problem, we could give a formulation corresponding to a multi-commodity flow problem over fixed time. This approach is particularly well fitting with our goals, because a quickest flow expresses the need to run a circuit as fast as possible, while a flow over fixed time brings a side interest into the minimization of resource usage, which is clearly a \textit{desideratum}, but still secondary to the overall running-time.

(ii) Constraints \eqref{c5} and \eqref{c6} give the possibility to consider more interesting solutions. In fact through efficient circuit manipulation -- see predicate $\mathcal{A}$ --, we managed to gather logically sequenced telegates within the same time step, achieving quasi-parallelism.



(iii) We built our model step by step, each of which rigorously explained. The result is an highly modular work. For example, if one can consider only circuits where operations can all commute each other, formulation \eqref{mcf} is enough and approximation bounds are available. Instead, when considering any circuit, one can easily shape the extra constraints of formulation \eqref{dqcc}. Consider, for example, the quasi-parallelism relation $\shortparallel$, we characterized it as the predicate $\mathcal{A}$. By just extending the way $\mathcal{A}$ works, the space of good solutions gets larger.

(iv) Since we modeled the problem as a network flow problem, one can also exploit the huge related literature to get inspiration in the way of tackling the problem. In next sub-section, we do an extensive discussion in this sense.

\subsection{Tackling the problem}
Formulation \eqref{dqcc} is a particular case of \texttt{MCF}$_d$, as it slightly recedes from the standard formulation.
As expected, the problem is still intractable. To understand that, consider this simple scenario: an instance $[k]$ with $k=2$ such that $1\shortparallel2$. We can restate the problem as follows: assert if there exists a solution at first time step. If not, just put operation $2$ at second time step. Unfortunately, asserting if such a solution exists is \texttt{NP}-hard. Indeed, in \cite{EveItaSha-75}, authors proved the hardness of such a decision problem, even for single capacity edges. Therefore, it is reasonable to look for approximations of \texttt{DQCC}.

To this aim, we think a good line of research would be to follow a common technique for tackling $\texttt{MCF}_d$: the \textit{time-expansion} \cite{ForFul-58}. Namely, a re-definition of the instance graph, from $\mathcal{Q}$ to a new graph $\mathcal{Q}_d$.
Such a technique is useful because, instead of tackling $\texttt{MCF}_d$ over $\mathcal{Q}$, one can tackle its static version $\texttt{MCF}$ over $\mathcal{Q}_d$. Let us introduce it formally for our scenario.

A time-expansion of $\mathcal{Q}$ is a graph $\mathcal{Q}_d = (P_d, R_d/\star)$. Accordingly to this criterion, an edge $(P_i,P_j)\in R/\star$ taking discrete travel time $\theta$ would translate into directed edges $(P_i(\tau),P_j(\tau + \theta)), (P_j(\tau),P_i(\tau + \theta)) \in R_d/\star$, with a shared constraint on the capacity. Nevertheless, edges in $\mathcal{Q}$ are assumed to have null travel time. Hence, a time-expansion of $\mathcal{Q}$ is particularly efficient, since one just needs to introduce a repetition of $\mathcal{Q}$ for each time-step $\tau$, which we refer to as $\mathcal{Q}(\tau) = (P(\tau), R(\tau)/\star)$. As consequence, time-dependent sets $P^C(\tau)$ and $P^T(\tau)$ replace $P^C$ and $P^T$.
We keep using $P^C$ and $P^T$ as the nodes encoding the commodities, non-localized in time. For each $i$ and $\tau$, we introduce edges $(P_i^C, P_i^C(\tau))$ and $(P_i^T(\tau), P_i^T)$, both with unit capacity.

Since only integral flow are allowed and the demand is exactly 1, for any operation $i$, only one of the edges $\{(P_i^C, P_i^C(\tau))\}_{\tau}$ -- as well as only one in $\{(P_i^T(\tau), P_i^T)\}_{\tau}$ -- will have a non-zero flow.

Now that we gave a first intuitive way to encode the sources of the problem, let us optimize it. Notice that operation $1$ can always run at time $1$, and it is a waste of time and space considering other options. As consequence, for operation $1$, we only introduce $(P_1^C, P_1^C(1))$ and $(P_1^T(1), P_1^T)$. This extends to any operation, which can always run in a time between $1$ and $\text{min}\{i,d\}$, by assuming that a solution exists with time horizon $d$. Therefore, for each operation $i$, we introduce the sets of edges $\{(P_i^C, P_i^C(\tau))\ :\ \forall 1\leq\tau\leq\text{min}\{i,d\}\}$ and $\{(P_i^T(\tau), P_i^T)\ :\ \forall 1\leq\tau\leq\text{min}\{i,d\}\}$.

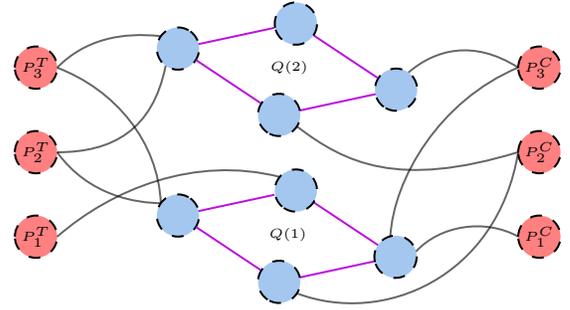
\begin{figure}
    \centering

\tikzset{every picture/.style={line width=0.75pt}} 

\begin{tikzpicture}[x=0.8pt,y=0.8pt,yscale=-1,xscale=1]

\draw [color={rgb, 255:red, 189; green, 16; blue, 224 }  ,draw opacity=1 ][line width=0.75]    (356.41,482.21) -- (320.47,489.97) ;
\draw  [color={rgb, 255:red, 0; green, 0; blue, 0 }  ,draw opacity=1 ][fill={rgb, 255:red, 74; green, 144; blue, 226 }  ,fill opacity=0.5 ][dash pattern={on 4.5pt off 4.5pt}] (305.37,500.46) .. controls (300.73,497.47) and (299.39,491.28) .. (302.38,486.64) .. controls (305.38,482) and (311.57,480.66) .. (316.21,483.66) .. controls (320.85,486.65) and (322.19,492.84) .. (319.19,497.48) .. controls (316.2,502.12) and (310.01,503.46) .. (305.37,500.46) -- cycle ;
\draw [color={rgb, 255:red, 189; green, 16; blue, 224 }  ,draw opacity=1 ][line width=0.75]    (357.57,474.68) -- (326.95,454.33) ;
\draw  [color={rgb, 255:red, 0; green, 0; blue, 0 }  ,draw opacity=1 ][fill={rgb, 255:red, 74; green, 144; blue, 226 }  ,fill opacity=0.5 ][dash pattern={on 4.5pt off 4.5pt}] (312.24,456.48) .. controls (308.03,452.92) and (307.5,446.61) .. (311.07,442.39) .. controls (314.63,438.17) and (320.94,437.64) .. (325.16,441.21) .. controls (329.38,444.78) and (329.9,451.09) .. (326.34,455.3) .. controls (322.77,459.52) and (316.46,460.05) .. (312.24,456.48) -- cycle ;
\draw  [color={rgb, 255:red, 0; green, 0; blue, 0 }  ,draw opacity=1 ][fill={rgb, 255:red, 74; green, 144; blue, 226 }  ,fill opacity=0.5 ][dash pattern={on 4.5pt off 4.5pt}] (356.12,482.64) .. controls (354.77,477.28) and (358.01,471.84) .. (363.36,470.49) .. controls (368.71,469.13) and (374.15,472.37) .. (375.51,477.72) .. controls (376.87,483.07) and (373.63,488.51) .. (368.27,489.87) .. controls (362.92,491.23) and (357.48,487.99) .. (356.12,482.64) -- cycle ;
\draw [color={rgb, 255:red, 189; green, 16; blue, 224 }  ,draw opacity=1 ][line width=0.75]    (272.64,458.57) -- (308.58,450.84) ;
\draw  [color={rgb, 255:red, 0; green, 0; blue, 0 }  ,draw opacity=1 ][fill={rgb, 255:red, 74; green, 144; blue, 226 }  ,fill opacity=0.5 ][dash pattern={on 4.5pt off 4.5pt}] (268.79,452.54) .. controls (273.28,455.76) and (274.3,462.01) .. (271.08,466.49) .. controls (267.86,470.98) and (261.61,472) .. (257.12,468.78) .. controls (252.64,465.56) and (251.61,459.31) .. (254.84,454.82) .. controls (258.06,450.34) and (264.31,449.31) .. (268.79,452.54) -- cycle ;
\draw [color={rgb, 255:red, 189; green, 16; blue, 224 }  ,draw opacity=1 ][line width=0.75]    (271.47,466.1) -- (302.08,486.48) ;
\draw [color={rgb, 255:red, 189; green, 16; blue, 224 }  ,draw opacity=1 ][line width=0.75]    (356.41,403.21) -- (320.47,410.97) ;
\draw  [color={rgb, 255:red, 0; green, 0; blue, 0 }  ,draw opacity=1 ][fill={rgb, 255:red, 74; green, 144; blue, 226 }  ,fill opacity=0.5 ][dash pattern={on 4.5pt off 4.5pt}] (305.37,421.46) .. controls (300.73,418.47) and (299.39,412.28) .. (302.38,407.64) .. controls (305.38,403) and (311.57,401.66) .. (316.21,404.66) .. controls (320.85,407.65) and (322.19,413.84) .. (319.19,418.48) .. controls (316.2,423.12) and (310.01,424.46) .. (305.37,421.46) -- cycle ;
\draw  [color={rgb, 255:red, 0; green, 0; blue, 0 }  ,draw opacity=1 ][fill={rgb, 255:red, 74; green, 144; blue, 226 }  ,fill opacity=0.5 ][dash pattern={on 4.5pt off 4.5pt}] (360.25,409.24) .. controls (355.77,406.02) and (354.74,399.78) .. (357.96,395.29) .. controls (361.18,390.8) and (367.43,389.77) .. (371.91,392.99) .. controls (376.4,396.21) and (377.43,402.46) .. (374.21,406.95) .. controls (370.99,411.44) and (364.74,412.46) .. (360.25,409.24) -- cycle ;
\draw [color={rgb, 255:red, 189; green, 16; blue, 224 }  ,draw opacity=1 ][line width=0.75]    (357.57,395.68) -- (326.95,375.33) ;
\draw  [color={rgb, 255:red, 0; green, 0; blue, 0 }  ,draw opacity=1 ][fill={rgb, 255:red, 74; green, 144; blue, 226 }  ,fill opacity=0.5 ][dash pattern={on 4.5pt off 4.5pt}] (308.9,371.82) .. controls (307.81,366.41) and (311.31,361.13) .. (316.73,360.04) .. controls (322.14,358.95) and (327.41,362.46) .. (328.51,367.87) .. controls (329.6,373.29) and (326.09,378.56) .. (320.68,379.65) .. controls (315.26,380.74) and (309.99,377.24) .. (308.9,371.82) -- cycle ;
\draw [color={rgb, 255:red, 189; green, 16; blue, 224 }  ,draw opacity=1 ][line width=0.75]    (272.64,379.57) -- (308.58,371.84) ;
\draw  [color={rgb, 255:red, 0; green, 0; blue, 0 }  ,draw opacity=1 ][fill={rgb, 255:red, 74; green, 144; blue, 226 }  ,fill opacity=0.5 ][dash pattern={on 4.5pt off 4.5pt}] (268.79,373.54) .. controls (273.28,376.76) and (274.3,383.01) .. (271.08,387.49) .. controls (267.86,391.98) and (261.61,393) .. (257.12,389.78) .. controls (252.64,386.56) and (251.61,380.31) .. (254.84,375.82) .. controls (258.06,371.34) and (264.31,370.31) .. (268.79,373.54) -- cycle ;
\draw [color={rgb, 255:red, 189; green, 16; blue, 224 }  ,draw opacity=1 ][line width=0.75]    (271.47,387.1) -- (302.08,407.48) ;
\draw  [color={rgb, 255:red, 0; green, 0; blue, 0 }  ,draw opacity=1 ][fill={rgb, 255:red, 255; green, 0; blue, 0 }  ,fill opacity=0.5 ][dash pattern={on 4.5pt off 4.5pt}] (423.82,390.65) .. controls (423.83,385.13) and (428.32,380.66) .. (433.84,380.68) .. controls (439.37,380.69) and (443.83,385.18) .. (443.82,390.71) .. controls (443.8,396.23) and (439.31,400.69) .. (433.79,400.68) .. controls (428.27,400.66) and (423.8,396.17) .. (423.82,390.65) -- cycle ;
\draw  [color={rgb, 255:red, 0; green, 0; blue, 0 }  ,draw opacity=1 ][fill={rgb, 255:red, 255; green, 0; blue, 0 }  ,fill opacity=0.5 ][dash pattern={on 4.5pt off 4.5pt}] (423.82,430.65) .. controls (423.83,425.13) and (428.32,420.66) .. (433.84,420.68) .. controls (439.37,420.69) and (443.83,425.18) .. (443.82,430.71) .. controls (443.8,436.23) and (439.31,440.69) .. (433.79,440.68) .. controls (428.27,440.66) and (423.8,436.17) .. (423.82,430.65) -- cycle ;
\draw  [color={rgb, 255:red, 0; green, 0; blue, 0 }  ,draw opacity=1 ][fill={rgb, 255:red, 255; green, 0; blue, 0 }  ,fill opacity=0.5 ][dash pattern={on 4.5pt off 4.5pt}] (423.82,470.65) .. controls (423.83,465.13) and (428.32,460.66) .. (433.84,460.68) .. controls (439.37,460.69) and (443.83,465.18) .. (443.82,470.71) .. controls (443.8,476.23) and (439.31,480.69) .. (433.79,480.68) .. controls (428.27,480.66) and (423.8,476.17) .. (423.82,470.65) -- cycle ;
\draw [color={rgb, 255:red, 0; green, 0; blue, 0 }  ,draw opacity=0.6 ]   (423.82,390.65) .. controls (402.33,377.01) and (384.33,382.34) .. (371.91,392.99) ;
\draw [color={rgb, 255:red, 0; green, 0; blue, 0 }  ,draw opacity=0.6 ]   (423.82,390.65) .. controls (386.5,406.27) and (365,443.27) .. (363.36,470.49) ;
\draw [color={rgb, 255:red, 0; green, 0; blue, 0 }  ,draw opacity=0.6 ]   (423.82,430.65) .. controls (414,490) and (360,513) .. (319.19,497.48) ;
\draw [color={rgb, 255:red, 0; green, 0; blue, 0 }  ,draw opacity=0.6 ]   (423.82,430.65) .. controls (367,448) and (342,436) .. (319.19,418.48) ;
\draw [color={rgb, 255:red, 0; green, 0; blue, 0 }  ,draw opacity=0.6 ]   (423.82,470.65) .. controls (407.67,461.68) and (386.33,465.68) .. (375.51,477.72) ;
\draw  [color={rgb, 255:red, 0; green, 0; blue, 0 }  ,draw opacity=1 ][fill={rgb, 255:red, 255; green, 0; blue, 0 }  ,fill opacity=0.5 ][dash pattern={on 4.5pt off 4.5pt}] (185.82,390.65) .. controls (185.83,385.13) and (190.32,380.66) .. (195.84,380.68) .. controls (201.37,380.69) and (205.83,385.18) .. (205.82,390.71) .. controls (205.8,396.23) and (201.31,400.69) .. (195.79,400.68) .. controls (190.27,400.66) and (185.8,396.17) .. (185.82,390.65) -- cycle ;
\draw  [color={rgb, 255:red, 0; green, 0; blue, 0 }  ,draw opacity=1 ][fill={rgb, 255:red, 255; green, 0; blue, 0 }  ,fill opacity=0.5 ][dash pattern={on 4.5pt off 4.5pt}] (185.82,430.65) .. controls (185.83,425.13) and (190.32,420.66) .. (195.84,420.68) .. controls (201.37,420.69) and (205.83,425.18) .. (205.82,430.71) .. controls (205.8,436.23) and (201.31,440.69) .. (195.79,440.68) .. controls (190.27,440.66) and (185.8,436.17) .. (185.82,430.65) -- cycle ;
\draw  [color={rgb, 255:red, 0; green, 0; blue, 0 }  ,draw opacity=1 ][fill={rgb, 255:red, 255; green, 0; blue, 0 }  ,fill opacity=0.5 ][dash pattern={on 4.5pt off 4.5pt}] (185.82,470.65) .. controls (185.83,465.13) and (190.32,460.66) .. (195.84,460.68) .. controls (201.37,460.69) and (205.83,465.18) .. (205.82,470.71) .. controls (205.8,476.23) and (201.31,480.69) .. (195.79,480.68) .. controls (190.27,480.66) and (185.8,476.17) .. (185.82,470.65) -- cycle ;
\draw [color={rgb, 255:red, 0; green, 0; blue, 0 }  ,draw opacity=0.6 ]   (254.84,375.82) .. controls (236.5,373.75) and (219,377.25) .. (205.82,390.71) ;
\draw [color={rgb, 255:red, 0; green, 0; blue, 0 }  ,draw opacity=0.6 ]   (254.84,454.82) .. controls (253.5,425.75) and (233.5,401.25) .. (205.82,390.71) ;
\draw [color={rgb, 255:red, 0; green, 0; blue, 0 }  ,draw opacity=0.6 ]   (311.07,442.39) .. controls (270.84,432.11) and (231.5,448.25) .. (205.82,470.71) ;
\draw [color={rgb, 255:red, 0; green, 0; blue, 0 }  ,draw opacity=0.6 ]   (257.12,389.78) .. controls (251,420.27) and (232,430.75) .. (205.82,430.71) ;
\draw [color={rgb, 255:red, 0; green, 0; blue, 0 }  ,draw opacity=0.6 ]   (254.84,454.82) .. controls (244,454.75) and (221,451.75) .. (205.82,430.71) ;

\draw (433.82,390.68) node  [font=\tiny]  {$P_{3}^{C}$};
\draw (433.82,430.68) node  [font=\tiny]  {$P_{2}^{C}$};
\draw (433.82,470.68) node  [font=\tiny]  {$P_{1}^{C}$};
\draw (304.67,465.33) node [anchor=north west][inner sep=0.75pt]  [font=\tiny]  {$Q( 1)$};
\draw (305.33,386) node [anchor=north west][inner sep=0.75pt]  [font=\tiny]  {$Q( 2)$};
\draw (195.82,390.68) node  [font=\tiny]  {$P_{3}^{T}$};
\draw (195.82,430.68) node  [font=\tiny]  {$P_{2}^{T}$};
\draw (195.82,470.68) node  [font=\tiny]  {$P_{1}^{T}$};

\end{tikzpicture}
    \caption{Time-expanded graph of $4$ processors, for an instance $[k]$ with $k=3$ and time horizon $d=2$.}
    \label{fig:time_expanded}
	\hrulefill
\end{figure}
Figure \ref{fig:time_expanded} shows the final graph for instance $[k]$ with $k=3$, time horizon $d=2$ on an architecture with $4$ processors.

As said, the time expansion $\mathcal{Q}_d$ is a common way to tackle $\texttt{MCF}_d$ as a static flow problem and it is particularly efficient in our scenario. Indeed, we already pointed out in Section \ref{sec:layer} how formulations \eqref{mcf} and \eqref{dqcc} belong to a group of flow problems which are dynamic-static hybrid. For this reason we could model $\mathcal{Q}_d$ by simply introducing $d$ repetitions of $\mathcal{Q}$.

To the best of our knowledge, even if approximation algorithms for \texttt{MCF} \cite{SriSta-00, ChoCho-06} and variants \cite{CheKhaShe-06, CheKhaShe-04, Sri-97, ChaCheGup-07} have been extensively studied, there seems to be no proposal relatable to ours, modeling \texttt{DQCC}. More formally, no efficient reduction seems possible from our problem to standard formulations, while approximation algorithms proposed in literature usually rely on \texttt{LP}-relaxation, or on greedy criteria that don't fit with constraints \eqref{c5} and \eqref{c6}. Hence, further studies along this line would be useful to (i) place the problem within its most proper complexity class and to (ii) guarantee approximation ratio.


\section{Conclusion}
We addressed the compilation problem on distributed architectures. In line with literature, we assumed telegates as fundamental inter-processor operations. To mitigate their impact to computation, we modeled a minimization problem with the running-time as objective function. 
Even if the main interest was in minimizing the running-time, this highly depends on (i) resource usage and (ii) circuit manipulation. Hence, finding a rigorous model that efficiently consider all these parties can be really tricky. To overcome this, we exploit the wide literature on dynamic flows. Specifically, as done for quickest multi-commodity flows, we embedded an \texttt{ILP} solver within a binary search over time. Hence, even if the primary goal is to find short-time solutions, this passes through a solver addressing (i) and (ii). More in detail, the objective function of the \texttt{ILP} formulation is the resource usage (i), but the constraints are based on circuit manipulation (ii). In other words, we embedded an evaluation of equivalent circuits through (automated) circuit manipulation. Specifically, the evaluation  considers to gather telegates within the same time step. As expected, integrating circuit manipulation to the formulation improved the quality of the solution space, in terms of running-time. In fact, our proposal -- see predicate $\mathcal{A}$ introduced in Section \ref{sec:predicate} --, introduces many better solutions (in terms of running-time) than without it. To quantify that, we showed a group of circuits that, without $\mathcal{A}$, are forced to run in the worst running-time possible -- i.e. as many time steps as the number of remote operations --, while $\mathcal{A}$ achieves best possible solutions -- i.e. \texttt{E}-depth $1$ --.

\appendix
\subsection{Entanglement swap generalization}
\label{apx:path}
\begin{figure}[b]
    \centering
    \begin{quantikz}[thin lines,row sep={0.75cm,between origins}, column sep ={0.3cm}]
		&\gate[2]{\texttt{E}} & \qw & \qw & \qw & \qw & \gate[style={fill=red!10}]{\texttt{Z}^{b_1}} & \qw  & \qw & \qw  & \qw &\gate[style={fill=blue!10}]{\texttt{Z}^{b_3}} & \qw\\
		&& \ctrl{1} & \gate{\texttt{H}} & \meter[]{}\rstick[]{$b_1$} & \\
		&\gate[2]{\texttt{E}}& \targ{} & \qw & \meter{}\rstick[]{$b_2$}\\
		&& \qw & \qw &\qw & \qw &\gate[style={fill=red!10}]{\texttt{X}^{b_2}} & \ctrl{1} & \gate[]{\texttt{H}} & \meter{}\rstick[]{$b_3$}\\
		  && & &&&\gate[2]{\texttt{E}}& \targ{}  &\qw & \meter{}\rstick[]{$b_4$} &\\
		 & &  & & &&  \qw& \qw  & \qw & \qw & \qw & \gate[style={fill=blue!10}]{\texttt{X}^{b_4}} & \qw
	\end{quantikz}
	\caption{Naive implementation of a path with length $2$ as two entanglement swaps in sequence.
	}
	\label{fig:naive_swap}
\end{figure}
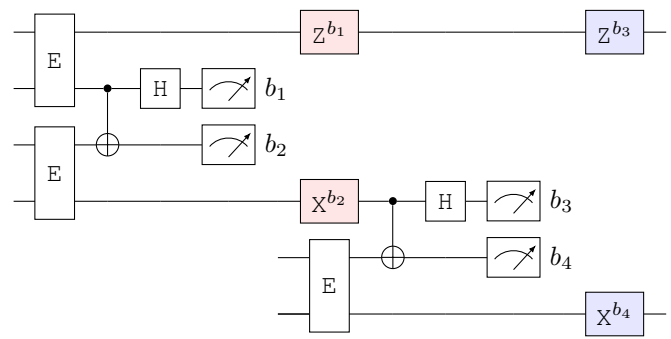
Within this section we show how to efficiently implement an entanglement path. In Section \ref{sec:swap}, we introduced the entanglement swap as a circuit of depth 5. We also claimed that such a depth is fixed when generalizing the entanglement swap to the entanglement path. To this aim, we give an inductive proof for such a statement, starting from the base case with entanglement path of length $2$. 


\begin{theorem}
\label{th:path}
    An entanglement path $\{P_{i_1}, P_{i_2}, \dots, P_{i_m}\}$ has an implementation with depth 5.
\end{theorem}
\begin{proof}
   Consider, as base case, that we want to create a path of length 2. Clearly, we could do that by just putting in strict sequence two entanglement swap. This way is showed in Figure \ref{fig:naive_swap}. The colored operators are the only ones we are going to optimize; since the others are independent and no optimization can be applied. The optimization is shown in Figure \ref{fig:base_path}. Specifically, circuit on the right of equation has post-processing composed by \texttt{Z}$^b$ on first qubit and \texttt{X}$^{\bar{b}}$ on last qubit. Furthermore, now the measurements are independent from other operations.
   
   By assuming that such a shape is preserved in the inductive step, we show that this transformation can be applied to any length -- see Figure \ref{fig:inductive_path}.
   This proves that we can always consider an entanglement path $\{P_{i_1}, P_{i_2}, \dots, P_{i_m}\}$ to have circuit depth 5.
\end{proof}

We just showed an efficient implementation for the entanglement path. Now we do one last step to exploit such a result and performing a generalized remote operation efficiently. Theorem \ref{th:path} allows us to assume that, to perform a remote operation by using a path of length $m$, the computing qubits interact only with two communications qubits and depend only by Pauli operations $\texttt{Z}^{b_1 \oplus b_3 \oplus\cdots \oplus b_{2m-1}}$ and $\texttt{X}^{b_2 \oplus b_4 \oplus\cdots \oplus b_{2m}}$. We can \textit{propagate} such operations as in the equivalence of Figure \ref{fig:inductive_rcx}. In this way the measurements are independent and the depth of the circuit has not increased.

\begin{figure*}[htbp]
    \centering
    \begin{quantikz}[thin lines,row sep={0.75cm,between origins}, column sep ={0.3cm}]
		& \gate[style={fill=red!10}]{\texttt{Z}^{b_1}} & \qw & \qw & \qw & \qw & \gate[style={fill=blue!10}]{\texttt{Z}^{b_3}} & \qw\\
		&\gate[style={fill=red!10}]{\texttt{X}^{b_2}} & \ctrl{1} & \gate[]{\texttt{H}} & \meter{}\rstick[]{$b_3$}\\
		& \qw  &\targ{} & \qw & \meter{}\rstick[]{$b_4$} &\\
		& \qw & \qw & \qw & \qw & \qw & \gate[style={fill=blue!10}]{\texttt{X}^{b_4}} & \qw
	\end{quantikz}
	$\equiv$
	\begin{quantikz}[thin lines,row sep={0.75cm,between origins}, column sep ={0.3cm}]
		& \qw & \qw & \qw & \gate[style={fill=violet!10}]{\texttt{Z}^{b_1\oplus b_3}}& \qw\\
		& \ctrl{1} & \gate[]{\texttt{H}} & \meter{}\rstick[]{$b_3$}\\
		& \targ{} & \qw & \meter{}\rstick[]{$b_4$} &\\
		& \qw & \qw &\qw & \gate[style={fill=violet!10}]{\texttt{X}^{b_2 \oplus b_4}} & \qw
	\end{quantikz}
	\caption{Base case equivalence of Theorem \ref{th:path}.}
	\label{fig:base_path}
\end{figure*}
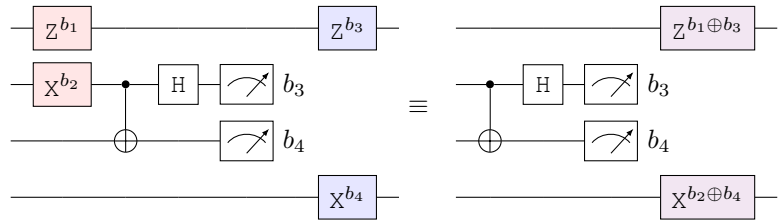
\begin{figure*}[htbp]
    \centering
    \begin{quantikz}[thin lines,row sep={0.75cm,between origins}, column sep ={0.3cm}]
		& \gate[style={fill=red!10}]{\texttt{Z}^{b_1\oplus b_3\oplus \cdots \oplus b_{2(m-1)-1}}}& \qw&\qw &\qw & \qw &\gate[style={fill=blue!10}]{\texttt{Z}^{b_{2m-1}}} &\\
		& \gate[style={fill=red!10}]{\ \texttt{X}^{b_2 \oplus b_4 \oplus\cdots \oplus b_{2(m-1)}}\ } & \ctrl{1}& \gate[]{\texttt{H}} & \meter{}\rstick[]{$b_{2m-1}$}\\
		& \qw & \targ{} & \qw & \meter{}\rstick[]{$b_{2m}$}\\
		& \qw & \qw &\qw &\qw &\qw &\gate[style={fill=blue!10}]{\ \texttt{X}^{b_{2m}}\ } &
	\end{quantikz}
	$\equiv$
	\begin{quantikz}[thin lines,row sep={0.75cm,between origins}, column sep ={0.3cm}]
		& \qw&\qw &\qw &\gate[style={fill=violet!10}]{\texttt{Z}^{b_1\oplus b_3\oplus \cdots \oplus b_{2m-1}}} &\\
		& \ctrl{1}& \gate[]{\texttt{H}} & \meter{}\rstick[]{$b_{2m-1}$}\\
		& \targ & \qw & \qw & \meter{}\rstick[]{$b_{2m}$}\\
		& \qw & \qw &\qw &\gate[style={fill=violet!10}]{\ \texttt{X}^{b_2 \oplus b_4 \oplus\cdots \oplus b_{2m}}\ } &
	\end{quantikz}
	\caption{Inductive step of Theorem \ref{th:path}.}
	\label{fig:inductive_path}
\end{figure*}
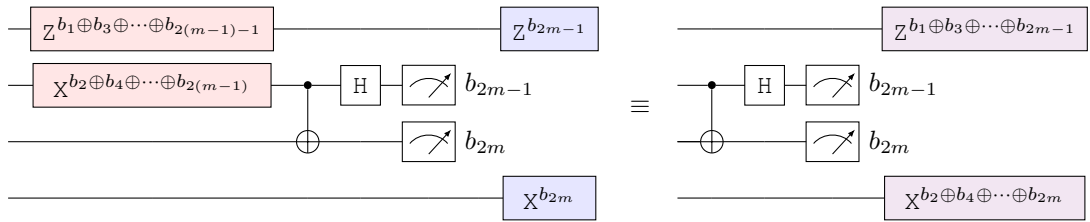
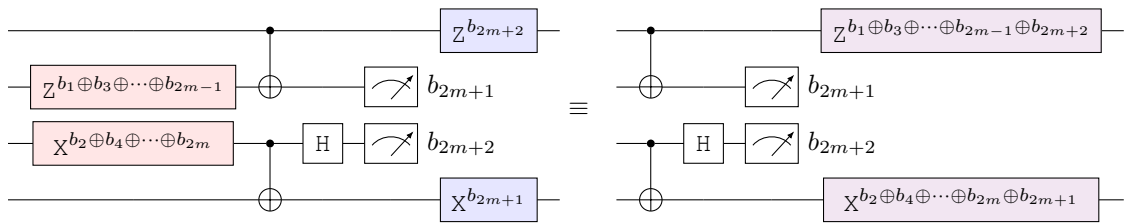
\begin{figure*}
    \centering
    \begin{quantikz}[thin lines,row sep={0.75cm,between origins}, column sep ={0.3cm}]
		& \qw& \ctrl{1} & \qw &\qw & \gate[style={fill=blue!10}]{\texttt{Z}^{b_{2m+2}}} &  \qw \\
		&\gate[style={fill=red!10}]{\texttt{Z}^{b_1 \oplus b_3 \oplus\cdots \oplus b_{2m-1}}}& \targ{} & \qw & \meter{} \rstick{$b_{2m + 1}$} \\
		&\gate[style={fill=red!10}]{\ \texttt{X}^{b_2 \oplus b_4 \oplus\cdots \oplus b_{2m}}\ }& \ctrl{1} & \gate{\texttt{H}} & \meter{} \rstick{$b_{2m + 2}$}\\
		&\qw& \targ{} & \qw & \qw & \gate[style={fill=blue!10}]{\texttt{X}^{b_{2m+1}}} & \qw
	\end{quantikz}
	$\equiv$
	\begin{quantikz}[thin lines,row sep={0.75cm,between origins}, column sep ={0.3cm}]
		& \ctrl{1} & \qw &\qw & \gate[style={fill=violet!10}]{\texttt{Z}^{b_1 \oplus b_3 \oplus\cdots \oplus b_{2m-1} \oplus b_{2m+2}}} &  \qw \\
		& \targ{} & \qw & \meter{} \rstick{$b_{2m + 1}$} \\
		& \ctrl{1} & \gate{\texttt{H}} & \meter{} \rstick{$b_{2m + 2}$}\\
		& \targ{} & \qw & \qw & \gate[style={fill=violet!10}]{\ \texttt{X}^{b_2 \oplus b_4 \oplus\cdots \oplus b_{2m} \oplus b_{2m+1}}\ } & \qw
	\end{quantikz}
	\caption{Final equivalence for the generalized remote operation.}
	\label{fig:inductive_rcx}
\end{figure*}


\end{document}